\documentclass[12pt]{article}
\usepackage[margin=1in]{geometry}
\usepackage[T1]{fontenc}
\usepackage{lmodern,microtype,amsmath,amssymb,amsthm,booktabs,array,longtable,graphicx}
\usepackage[colorlinks=true,linkcolor=blue,urlcolor=blue,citecolor=blue]{hyperref}
\newif\ifblind
\ifdefined\BlindVersion\blindtrue\else\blindfalse\fi
\ifblind
\hypersetup{pdftitle={Shared Models, Selective Trading, and Order Flow},pdfauthor={}}
\else
\hypersetup{pdftitle={Shared Models, Selective Trading, and Order Flow},pdfauthor={Victoria Ruojie Li and Arka Prava Bandyopadhyay}}
\fi
\usepackage[round,authoryear]{natbib}
\usepackage{tikz,fvextra}
\usetikzlibrary{arrows.meta,positioning}
\usepackage[font=small,labelfont=bf]{caption}
\usepackage{setspace}
\newtheorem{proposition}{Proposition}
\newtheorem{corollary}{Corollary}
\newcommand{\E}{\mathbb{E}}
\newcommand{\Prb}{\mathbb{P}}
\title{Shared Models, Selective Trading,\\and Order Flow}
\ifblind
\author{}
\else
\author{Victoria Ruojie Li \and Arka Prava Bandyopadhyay}
\fi
\date{November 2026}
\begin{document}
\maketitle
\begin{abstract}
We study whether the diversity of language models installed in a trading population carries over to the orders that reach the market. In simulated markets with five agents from each of three open-weight model families, we change only how a corporate announcement is presented and estimate each family's share of active orders. In two matched experiments, explanatory layoff text multiplies Mistral's participation by 5.3 (from 14.0\% to 74.7\% of opportunities) and raises its share of announcement-round orders from 17.5\% to 57.4\%. In a preregistered study of six new events, adding one favorable and one unfavorable evaluation to identical facts changes active composition in a main sample and in an independently worded replication (average e-values of $6.3\times10^{25}$ and $6.7\times10^{40}$, against a threshold of 20). Within an event, the change reallocates on average 21\% of active decisions across families in the main sample and 14\% in the replication. The same reallocation appears under both wordings in three of the six events; in the other three, the families that gain differ. Evaluative text also lowers every family's average participation and lowers executed net flow in all twelve event-sample cells, and family contributions mostly reinforce one another. In the financing event of the first study, active composition changed substantially while net flow barely moved, because family responses offset one another. In a stylized market-clearing benchmark, price accuracy depends on both the number and the composition of active traders. The experiments do not identify why the models respond to evaluations, and they do not estimate effects on price accuracy.
\end{abstract}

\section*{Introduction}
Financial firms increasingly use language models to read news and to generate or screen trading decisions. When several models are in use, their differences are expected to limit correlated trading. That expectation concerns the installed population of models. Prices, however, aggregate the orders of those who choose to trade. This paper asks whether the diversity of installed models carries over to the orders that reach the market, and how the presentation of a news event affects the answer. We study the question in markets populated by language-model agents.

The question bears on how finance research and policy evaluate AI in markets. Evidence that a model extracts predictive information from news concerns its standalone performance \citep{gu,lopez2026}. Policy assessments of AI in finance emphasize concentration among model providers and the risk of correlated market behavior \citep{fsb2024}. Both assessments rely on which models are in use. Prices depend instead on which models' decisions become orders. If presentation changes which model families trade, a count of the models or providers in use does not measure the diversity of the demand that prices aggregate, and the same installed population can act as a diverse or as a concentrated one depending on how an event is described. Investors do not control presentation: issuers, the media and data vendors decide how an event is described.

The setting contains 15 agents, five from each of three model families, with heterogeneous mandates and endowments. Matched experiments vary the presentation of a financing announcement and of a workforce-reduction announcement. Separate assignments replace the mixed population with a homogeneous one. Thirty configurations per contrast hold starting conditions and action-order schedules fixed, while model-generation draws remain stochastic. Thick background liquidity limits price movement, so the outcomes are submitted decisions rather than price discovery.

The first result is that presentation changes active representation by large amounts. In the layoff event, explanatory wording raises Mistral's share of announcement-round orders from 17.5\% to 57.4\%, although Mistral is one-third of the agents in both arms. The shift of 39.9 percentage points has an exploratory paired-configuration bootstrap interval of [32.8, 47.5], and dropping any one configuration leaves it between 38.8 and 41.3 points. In the financing event, the original wording raises Qwen's active share by 48.1 points and lowers Llama's by 61.0 points relative to a factual version of the same announcement. Proposition~\ref{prop:installed} gives the general conditions. Installed shares describe active flow only when families participate at equal rates, and framing changes active composition only when it changes participation in different proportions across families; a framing that doubles every family's willingness to trade changes volume but not composition. Study~1 departs sharply from both conditions. Explanatory layoff text multiplies Mistral's participation by 5.3, while Qwen's and Llama's fall by 17\% and 15\%, and across the four arms installed shares misstate active composition by 18 to 45 percentage points of total variation.

The second result concerns net order flow, where the two events differ. Layoff wording reverses the sign of net submitted orders, from 54 net sells to 55 net buys. Financing wording leaves them almost unchanged, at 29 and 26 net sells. Proposition~\ref{prop:netflow} explains the difference. The change in net orders is the sum of family contributions, and these contributions can reinforce or offset one another. In the layoff event they reinforce: Mistral adds 91 net buys and Qwen sells 20 fewer on net, and the change in net orders is 96\% of the sum of absolute family contributions. In the financing event they offset: Qwen adds 68 net sells, while Mistral adds 28 net buys and Llama sells 37 fewer, so the net change is three orders, or 2\% of that sum. In general, a presentation change moves net flow most when it draws families on one side of the market into trading while families on the other side withdraw. It leaves net flow unchanged when it draws in families on both sides, or when it reallocates representation among families that trade in the same direction. In active-share terms, the composition contribution equals the share of active decisions reallocated times the gap in conditional direction between the families that gain and the families that lose representation. In the layoff event, 40\% of active decisions move from families that sell on balance to a family that only buys, and this term accounts for 108 of the 109-order change.

A preregistered study tests the composition result prospectively in six new corporate events with a responsive liquidity provider (Section~\ref{sec:study2}). Both arms state the same facts, and the evaluative arm adds one favorable and one unfavorable evaluation of them. With 1,500 configuration pairs per event, the confirmatory test rejects the null of unchanged active composition in the main sample ($E_{H2}=6.3\times10^{25}$, against a threshold of 20) and, separately, in a replication with independently worded texts ($E_{H2}=6.7\times10^{40}$). The test is global, and the event-level estimates show where the shifts occur. Within an event, evaluative framing reallocates on average a fifth of active decisions across families in the main sample and a seventh in the replication. The same reallocation recurs across the two wordings in three of the six events: in the customer contract, the product redesign and the division sale, the same families gain and lose in the same direction. In the division sale, for example, Qwen gains 18 and 21 points of active share and Mistral loses 19 and 17 points. In the other three events the families that gain differ across wordings. The evidence therefore supports the claim that evaluative framing changes active composition, and it does not support a claim about which families gain, which depends on the event and its wording. The most consistent result is the fall in net buying. Because presentation is the treatment, sensitivity to the text is the object of study, and a uniform prompt-sensitivity effect would leave active composition unchanged (Proposition~\ref{prop:installed}(b)); the test rejects that null.

The flow consequences are more uniform. On average across events every family trades less, buy orders fall by more than half, and executed net flow is lower under evaluative text in all six events of both samples, falling from about six net shares bought per market to roughly zero in the main sample. Family contributions to this change mostly reinforce one another: the reinforcement ratio of Proposition~\ref{prop:netflow} is at least 0.76 in eleven of the twelve event-sample cells, and the exception is the replication patent licence (0.47). This is the case of Proposition 2 in which selection moves price pressure. The design does not show that the change in composition causes the fall in net flow. An exploratory accounting attributes most of the change to fewer active decisions and to less net buying per active decision, and its ranking of composition against within-family terms depends on the level of aggregation (Section~\ref{sec:study2}). A later extension with known asset values would make decision regret and value deviation measurable (Appendix~\ref{app:knownvalue}).

The third result connects this accounting to prices, in a benchmark that we do not estimate. In a linear market-clearing model in which every active agent has the same demand sensitivity, the price error loads on each family's valuation error in proportion to its active share (Proposition~\ref{prop:weights}). Framing then affects the price error through two channels that mirror the flow accounting: the number of active agents, which sets how much weight model errors receive relative to other traders, and active composition, which sets which errors receive it. A reallocation that leaves net flow unchanged can therefore still change which errors enter the price. Selection into active demand can improve or worsen accuracy, and the sign depends on the covariance between selection and valuation errors (Proposition~\ref{prop:selection}); installed diversity, an active concentration index and standalone model quality do not determine it. The financing contrast illustrates the measurement point: net order flow shows almost no response while the families supplying demand change substantially. The experiments contain no independent fundamental value, so they do not estimate any effect on price accuracy.

Homogeneous populations provide a boundary case. Under the original financing text, all-Mistral markets submit only buys and all-Qwen markets only sells in every announcement window, whereas four of 30 mixed-control windows are nonempty and fully one-sided. This comparison changes family identity and pre-event history together. It is a total population-assignment effect, not a quality-held-constant estimate of a concentration effect. Follow-up and hold-rationale evidence further show that active direction cannot be read as the opinion of every agent in a family.

Our contribution is to show, in two matched experiments and a preregistered, replicated test, that the presentation of an event changes which model families supply orders within a fixed installed mix. Installed diversity therefore does not ensure stable active representation. We also characterize when such participation changes move net order flow and the weights on which prices load, and we report in which events the composition shifts recur across wordings and in which they do not. The findings concern three open-weight model families acting as trading agents in synthetic markets with a specified supply rule. They do not identify why the models respond to evaluations, and they do not establish effects on price efficiency. Three literatures are close, and each studies a different object. First, LLM families differ in markets. \citet{struski} run double auctions in which each market contains a single model, and the models differ in how often they trade and how far they cross the spread. \citet{henning} and \citet{saxena} run asset markets with single-model and mixed-model populations, and their outcomes are prices, bubbles and profits. \citet{ouyangsui} run single-model and two-model markets and add instructions that amplify or suppress specific heuristics. These change bubble size, but they instruct the agents rather than change the text of an event. \citet{lopeztrade} shows that LLM agents follow assigned strategies in an order-book market. Second, LLM decisions respond to presentation. In a threshold-voting task, logically equivalent framings shift choices in most model families \citep{wangzhang}. In an oligopoly pricing game, injected market commentary moves the prices that LLM agents set, with family-specific sensitivity \citep{leepark}. Third, shared models can correlate decisions \citep{kleinberg,bommasani}. \citet{rosslo} find that the decisions of more capable models are more correlated with one another, in markets that each contain one model family at a participation share set by the experimenter. \citet{cataneo} documents provider differences in trading signals.

We are not aware of a study that holds a mixed population of model families fixed, changes only the text of an event and measures which families supply the active orders. The adjacent results do not settle this question. Differences in how often single models trade \citep{struski} suggest that installed shares can misstate active shares under a given text. They do not show whether a change of text moves families' participation in different proportions, which is the condition under which composition changes (Proposition~\ref{prop:installed}). In a pricing game every agent sets a price in every period, so family-specific price responses \citep{leepark} cannot separate a change in the decision to act from a change in direction. That separation determines net flow. A text can change each family's probability of acting and its direction conditional on acting, and the two responses can reinforce or offset each other in net order flow. In the financing contrast, Qwen's active share rises by 48 points while net orders change by three. Correlation results concern how similar the models' decisions are \citep{kleinberg,bommasani,rosslo}. Our margin is which models act, so a population can keep its installed diversity and still lose it in active flow. Study~1's bundles change interpretation, social cues and length together. Study~2 holds the facts fixed and adds evaluations, so it isolates evaluative elaboration of identical facts; without a length-matched control or, as in \citet{leepark}, a condition that adds neutral sentences, it does not separate evaluation from added length.

\paragraph{Relation to finance research.}
The paper connects AI information processing \citep{gu,lopez2026} to participation and aggregation. Media exposure and stale-news studies show why participation belongs in the analysis of financial information \citep{engelberg,tetlock}, and \citet{cheng} examines trading around ChatGPT outages. Our margin is representation within a fixed population of tools, not tool availability or human attention. Information diversity and endogenous responses can alter the link between better inputs and prices \citep{goldstein,banerjee}; our selection calculation makes a related measurement point in a much simpler demand system. The price-error benchmark is distinct from the return implications of machine forecast disagreement in \citet{bali}, as Section~\ref{sec:pricing} discusses.

Algorithmic trading can improve market quality \citep{hendershott}; our results concern which models supply orders, not whether algorithmic trading improves or harms market quality. The selection margin is also distinct from the reinforcement-learning collusion mechanism in \citet{dou}. In a replication of a classic herding experiment, LLM agents rely more on private information and herd less than the human professionals in the original study \citep{hansenlee}. Our question is not herding but which families act. Model information sets and leakage require explicit controls \citep{glasserman}. Synthetic events remove the need to recall the realized return of a named firm, but they do not remove learned financial associations. Table~\ref{tab:position} lists the closest comparisons, and Appendix~\ref{app:legacy} reports the auxiliary diagnostics.
\begin{table}[htbp]\centering\small\singlespacing
\caption{Related evidence and the active-representation margin}\label{tab:position}
\begin{tabular}{>{\raggedright\arraybackslash}p{.27\linewidth}>{\raggedright\arraybackslash}p{.29\linewidth}>{\raggedright\arraybackslash}p{.34\linewidth}}\toprule
Research & Existing object & Distinction in this paper\\\midrule
Lopez-Lira and Tang (2026), JFE; Gu et al. (2020), RFS & Information extraction and return prediction & Which available models supply active orders\\
Bali et al. (2026), RFS & Machine forecast dispersion and future returns & Active weights versus installed forecast diversity\\
Tetlock (2011), RFS; Engelberg and Parsons (2011), JF; Cheng et al. (2025), JAE & News exposure, repetition and AI availability & Presentation within a fixed mixed population\\
Goldstein and Yang (2015), JF; Banerjee et al. (2018), RFS & Information structure and equilibrium prices & Participation accounting and a conditional price benchmark\\
Hendershott et al. (2011), JF & Automation and liquidity & Flow composition is not itself market quality\\
Lopez-Lira (2025); Henning et al. (2025); Saxena et al. (2026), working papers & LLM markets; single versus mixed populations & Installed composition held fixed for the primary text contrast\\
Struski et al. (2026), working paper & Trade frequency and spread crossing differ across single-model markets & Text-induced changes in which families act within one mixed market\\
Lee and Park (2026); Wang and Zhang (2026); Cataneo (2026), working papers & Presentation-sensitive pricing and choice; provider-specific signals & Selection into submitted order flow; acting versus direction\\
Kleinberg and Raghavan (2021), PNAS; Ross et al. (2026), working paper & Correlated decisions from shared or more capable models; LLM share set by design & Endogenous participation by family at a fixed installed mix\\
Dou et al. (2025), working paper & Reinforcement learning, collusion and price efficiency & Non-strategic accounting of active representation\\\bottomrule
\end{tabular}
\end{table}

\section{Experimental environment}
\subsection{Agents, information, and market mechanics}
Each market contains 15 agents and runs for 25 rounds, indexed from zero. In the mixed condition, five agents use each of Llama 3.1 8B, Qwen 2.5 7B, and Mistral 7B. The implementation calls the corresponding Ollama model tags. A seed draws initial cash between 900 and 1,100, initial holdings from zero to five shares, and a mandate from four categories: modest growth with low risk, capital preservation, opportunistic trading with small risks allowed, and steady income. The seed also determines the action-order randomization.

The prompt supplies the mandate, the agent's cash and shares, the recent price, top bid and ask sizes, and the current news. It permits holding and requests a JSON action, limit price, quantity, and short reason. It does not supply other agents' identities, model assignments, or reasoning. Agents act sequentially and can therefore encounter market states affected by earlier actions. The absence of direct communication is not an assumption of statistical independence.

Background bid and ask orders are thick relative to the agents' permitted orders: the runner uses 300-share background liquidity at prices 99.5 and 100. Submitted quantities are capped at five shares and restricted by cash or inventory. Invalid or zero-quantity actions can become holds during cleaning. Throughout, an active decision means a cleaned buy or sell with positive quantity. It is not necessarily an execution, and order counts are not share volume.

Background orders are posted once at round zero and are not replenished. Incoming orders execute against the best crossing resting order at its limit price; stable sorting gives arrival priority within a price. Agent orders can therefore execute against one another inside the background spread. Across the 308 audited run files (including designated pilots and the all-Llama comparison), the maximum cumulative submitted quantity on either side in a run is 94 shares, below the initial depth of 300 on each side. Submitted quantity bounds background consumption from above because some agent orders match one another. All recorded last prices in this audited set lie between 99.5 and 100.0, including interior values. The narrow realized price band is an empirical property of these runs, rather than a general guarantee of a double auction.

No JSON parse failures occur in the 115,500 decisions in that audited set. This is a data-integrity denominator, not the estimation sample for any one result; pilots remain excluded from the main contrasts. Successful parsing does not imply that a requested action is feasible or survives cleaning.

We use two named temporal conventions. The \emph{event round} is exactly $t=5$ or $t=12$, used for participation and rationale counts. The \emph{event window} comprises that round and the following two rounds, used for the stated co-action and homogeneous-flow analyses. The original co-action statistic requires overlapping active rounds; later first-action diagnostics use the first active direction anywhere in the window. Section~2 specifies their distinct pair weighting. The exact prompt and recorded treatment texts are reproduced in Appendix~\ref{app:stimuli}.

The thick background liquidity constrains the price response. Small observed transaction-price movements within this structure do not identify price discovery, fundamental mispricing, or price impact in a market with elastic liquidity supply. Our aggregate outcomes concern the agents' submitted flow, excluding background orders.

\subsection{Experiments and matching}
Table~\ref{tab:design} distinguishes the experimental batches. The main reported contrasts use 30 seed configurations each. The original mixed financing batch and the later Design A original-text batch are distinct stochastic realizations; their counts should not be interchanged. The homogeneous comparison uses the later Design A mixed original-text control. Earlier two-seed pilots are excluded from the main estimation counts.

\begin{table}[htbp]\centering\small
\caption{Experimental contrasts and interpretation}\label{tab:design}
\begin{tabular}{p{0.22\linewidth}p{0.34\linewidth}p{0.34\linewidth}}
\toprule
Experiment & Contrast & Identified comparison in this setting\\\midrule
Original financing & Mixed population, announcement and follow-up nodes & Descriptive family responses and conditional co-action\\
Follow-up text & Interpretive, factual restatement, no-update text at $t=12$ & Response to the follow-up text bundle under matched starting conditions\\
Design A financing & Original versus de-emotionalized announcement at $t=5$ & Presentation-bundle effect within financing\\
Homogeneous arms & All Mistral or all Qwen versus mixed control & Population-composition intervention under original financing text\\
Layoff robustness & Factual versus explanatory announcement at $t=5$ & Presentation-bundle effect within a second event\\\bottomrule
\end{tabular}
\end{table}

The financing event describes an equity issue associated with funding relief and an approximately 30\% increase in share count. The original text adds interpretations and describes some traders buying while others sell. Design A removes the interpretive layer while retaining the event's core financing facts. The layoff event describes a 20\% workforce reduction, a 25\% operating-cost reduction, a longer cash runway, the loss of one of three development teams, and two quarters of sales growth below plan. The explanatory arm adds consequences and embedded accounts of traders' reactions.

These are \emph{bundled presentation treatments}. In particular, language about permanently reduced capacity or what other traders believe may convey substantive propositions as well as tone. A causal effect of assignment to the supplied text is not a separately identified effect of sentiment vocabulary, emotionality, or information-free framing. The factual follow-up in the original study also mentions both buy and sell orders, so it is not a perfectly social-cue-free control.

In the layoff study, odd seeds receive cash-first clause order and even seeds capacity-first order. The same order is used in both arms of each seed. This balances the two orders across the batch but does not randomize order independently of seed or provide a within-seed order contrast.

\subsection{Replication and stochastic variation}
Configuration (simulation) seeds are recorded for every run; they fix structural starting conditions and action-order schedules, not model-generation randomness. The agent code called Ollama with the server's default sampling settings and no explicit generation seed, so per-call generation seeds and the context length were neither set nor recorded. Thus, identical simulation seeds do not imply identical sampled language-model outputs or identical realized pre-event states across arms. The paired estimates should be read as differences under matched configurations with remaining sampling noise, rather than exact deterministic counterfactuals at a common realized history.

In text comparisons, pre-intervention differences reflect remaining generation variation under the common pre-intervention design. In composition comparisons, the family assignment applies from the beginning of the run, so differences in pre-announcement histories are themselves part of the treatment path. The latter contrast is a total population-assignment effect under matched initial conditions, not an effect holding realized event-time portfolios or prices fixed. The repeated announcement batches below provide an empirical check on directional stability under generation variation, without fixing that variation by design.

The archive records Ollama 0.31.2 for Design A and the homogeneous arms and 0.33.2 for the layoff batch; no server version was recorded for the original financing batches and their follow-up controls. Within-batch comparisons use the same recorded version. Cross-theme comparisons also change inference-stack version and stimulus construction, precluding a clean theme-by-wording interaction. The archived preregistration files establish the documented analysis intentions; this paper does not claim an independently time-stamped public registration. Study~2 fixes and records sampling settings, context length and per-call generation seeds (Appendix~\ref{app:study2}); these settings are a design choice, not a reconstruction of the historical environment.
\section{Outcomes and inference}
\subsection{Participation and conditional direction}
For agent $i$, seed $s$, round $t$, and treatment $z$, write the cleaned action as $A_{ist}(z)\in\{-1,0,1\}$, representing sell, hold, and buy. Let $D_{ist}=\mathbf{1}\{|A_{ist}|=1\}$. Announcement-round participation for family $f$ is
\begin{equation}
R_{fs}(z)=\sum_{i:F_i=f}D_{is5}(z).
\end{equation}
The paired mean effect is $\widehat\Delta_f=30^{-1}\sum_s[R_{fs}(1)-R_{fs}(0)]$. Each family contributes five decisions per seed at this round, so pooled announcement tables have a denominator of 150 per family and arm. Three-round windows have 450 decision opportunities. Mixing these denominators changes the estimand.

We also report $p_f(z)=\Prb(A=1\mid D=1,F=f,z)$, estimated from active decisions within the indicated round or window. This conditional probability is distinct from the unconditional buy rate. Because participation responds to treatment, a difference in $p_f$ compares selected active populations. It need not represent an individual-level change in direction among agents that would trade in both arms.

\subsection{Co-action and imbalance}
The later Design A and layoff analyses define an agent's window direction by its first active action in $W_5=\{5,6,7\}$. Within a seed, jointly active agent pairs are classified as same-family or cross-family, and agreement is the indicator that their first directions match. The original audit instead compares actions only in rounds when both agents act: it averages agreement over their shared active rounds, then equally over eligible pairs in each class within a seed. Pairs without an overlapping active round are excluded. These statistics should not be silently pooled across conventions. We report the original seed-mean gap separately from the later first-action diagnostics.

A seed-level same-minus-cross gap is defined only when both pair categories are observed. We report its computable-seed count. Pooled pair agreement weights seeds by their eligible pairs, whereas a mean of seed gaps weights computable seeds equally. They are different summaries; their numerical differences are not errors.

Let $B_s$ and $S_s$ count all submitted buy and sell decisions in the announcement window. Define
\begin{equation}
OS_s=\frac{|B_s-S_s|}{B_s+S_s},\qquad
FULL_s=\mathbf{1}\{B_s+S_s>0,\ B_sS_s=0\}.
\end{equation}
$OS_s$ is undefined when nobody trades. $FULL_s$ is zero for such a seed under the supplied analysis convention. Consequently, inference for continuous one-sidedness and the fully-one-sided indicator can have different denominators and p-values.

\subsection{Uncertainty and multiplicity}
The simulation seed, not an agent pair or an individual decision, is the unit for reported paired tests. Exact sign tests discard zero differences unless a tie-as-failure sensitivity is explicitly indicated. These tests concern the balance of positive and negative differences among nonzero cases, not directly the null that the mean effect is zero. Percentile intervals in the supplied scripts resample seeds 10,000 times and summarize uncertainty in mean contrasts under the chosen configuration distribution.

For layoff total participation, the documented primary alternative is one-sided. For family participation, the documented alternatives are two-sided, with each family assessed against a 0.01 threshold. The shipped script prints one-sided family p-values; our corrected table uses the two-sided values. The archive's rule is not itself a general familywise error guarantee at 0.01. Mistral's result also survives a Bonferroni correction for three family tests, so this issue does not change that conclusion. Other diagnostics are descriptive unless explicitly designated otherwise.

\section{Evidence}
\subsection{Active composition and aggregate direction}\label{sec:composition}
The three families each represent one-third of agent instances. Their shares of active decisions need not equal one-third. Let $n_f$ denote pooled announcement actions by family and let $\widetilde w_f=n_f/\sum_g n_g$. Define active-family concentration as $H_A=\sum_f\widetilde w_f^2$. Table~\ref{tab:active} is computed from the archived raw announcement logs. These are pooled arm-level quantities across seeds, not means of within-market indices.

\begin{table}[htbp]\centering\small\singlespacing
\caption{Pooled announcement composition and direction}\label{tab:active}
\begin{tabular}{llrrrr}\toprule
Event & Text & Active & $H_A$ & Buy share & Net orders\\\midrule
Financing & Original &129&0.390&0.388&$-29$\\
Financing & Factual &66&0.655&0.303&$-26$\\
Layoff & Explanatory &195&0.438&0.641&55\\
Layoff & Factual &120&0.441&0.275&$-54$\\\bottomrule
\end{tabular}
\par\smallskip\footnotesize Population-family concentration is $1/3$ in all arms. Net orders are buy decisions minus sell decisions, excluding background liquidity. This table does not measure volumes, executions, price errors, or offsetting orders within the same market.
\end{table}

Financing factual text concentrates active decisions in Llama, raising $H_A$ even though agent assignments are unchanged. Layoff wording produces a different pattern: $H_A$ is approximately unchanged while the dominant active family changes from Qwen to Mistral and the sign of pooled net flow reverses. A provider-share or active-family Herfindahl alone therefore misses the identity and direction of the represented readings.

For each arm $z$, let $B(z)=\sum_f\widetilde w_f(z)p_f(z)$ be the pooled buy share among active decisions. When all family-arm marginals are defined, write $\overline p_f=[p_f(1)+p_f(0)]/2$ and $\overline w_f=[\widetilde w_f(1)+\widetilde w_f(0)]/2$. Then
\begin{equation}\label{eq:accounting}
B(1)-B(0)=\underbrace{\sum_f\Delta\widetilde w_f\overline p_f}_{\text{active-composition component}}
+\underbrace{\sum_f\overline w_f\Delta p_f}_{\text{conditional-direction component}}.
\end{equation}
This exact identity symmetrizes the two possible sequences of changes; expanding either side proves equality. For explanatory minus factual layoff text, the two components are 0.3441 and 0.0219, summing to 0.3660. For original minus factual financing text, they are 0.2012 and $-0.1167$, summing to 0.0846. The latter conditional component is sensitive to Qwen's factual-arm marginal, estimated from only two active decisions.

These are ex post accounting terms. They do not hold the identity or potential direction of participants fixed and are not path-specific causal effects. Pooled participation is positive for every family-arm cell here; if a cell were empty, its conditional direction would be undefined and the corresponding decomposition could not be reported without an additional convention. We report no post hoc significance test for these components. Their role is to describe how the observed extensive and conditional-direction margins combine, using all 30 seeds without pretending that pooled trades are independent.

\begin{figure}[htbp]\centering
\includegraphics[width=\linewidth]{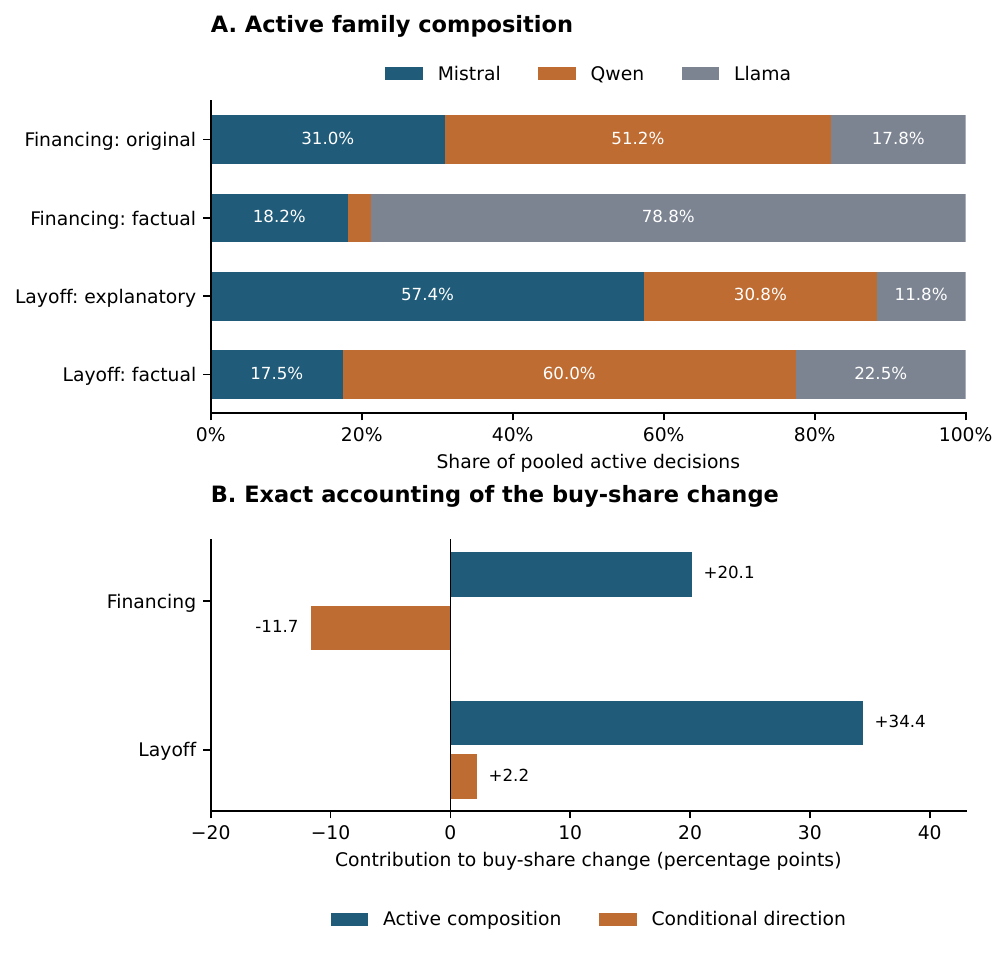}
\caption{Presentation changes the representation of model families}
\label{fig:composition}
\par\smallskip\footnotesize Panel A shows family shares of pooled announcement actions; the population share is one-third per family in every arm. Panel B decomposes the change in the pooled buy share using equation~\eqref{eq:accounting}; original minus factual for financing and explanatory minus factual for layoffs. These are descriptive pooled quantities, without causal-mediation interpretation or independent-trade inference.
\end{figure}

\subsection{A second event and its limits}
The layoff study uses a distinct event and two matched presentation arms. Table~\ref{tab:layoff} uses the event-round convention: explanatory Llama has 23 active decisions at $t=5$, whereas the audit reports 25 over the event window. These are consistent counts for different windows. At the announcement, total activity rises from 120 to 195, a 62.5\% increase.

\begin{table}[htbp]\centering
\caption{Layoff announcement decisions at $t=5$}\label{tab:layoff}
\begin{tabular}{llrrrr}\toprule
Text & Family & Buy & Sell & Hold & Active rate\\\midrule
Factual & Mistral &21&0&129&14.0\%\\
Factual & Qwen &5&67&78&48.0\%\\
Factual & Llama &7&20&123&18.0\%\\\addlinespace
Explanatory & Mistral &112&0&38&74.7\%\\
Explanatory & Qwen &9&51&90&40.0\%\\
Explanatory & Llama &4&19&127&15.3\%\\\bottomrule
\end{tabular}
\end{table}

\begin{table}[htbp]\centering\small
\caption{Paired layoff participation effects: explanatory minus factual}\label{tab:tests}
\begin{tabular}{lrrrrl}\toprule
Outcome & Mean & Positive & Negative & Ties & Exact $p$\\\midrule
Total &2.500&26&3&1&$7.62\times10^{-6}$ (one-sided)\\
Mistral &3.033&28&0&2&$7.45\times10^{-9}$ (two-sided)\\
Qwen &$-0.400$&8&12&10&0.503 (two-sided)\\
Llama &$-0.133$&10&13&7&0.678 (two-sided)\\\bottomrule
\end{tabular}
\par\smallskip\footnotesize Thirty paired seeds; mean in active decisions per seed. Bootstrap 95\% intervals: total $[1.60,3.37]$; Mistral $[2.53,3.50]$; Qwen $[-0.97,0.17]$; Llama $[-0.57,0.30]$. Ties dropped in exact tests.
\end{table}

Mistral supplies the aggregate increase. Qwen and Llama differences are not statistically resolved by the documented family tests. In the window-based directional criterion, a family is evaluable at 15 or more active decisions and near-unanimous if its buy fraction is at least 0.9 or at most 0.1. Three of six family-arm cells qualify: Mistral in both arms and Qwen in the factual arm. Qwen's explanatory buy share is exactly $9/60=0.15$, which fails the 0.1 threshold. Thus the directional hypothesis is only partially met.

Factual wording does not generally make Qwen inactive: Qwen has 72 announcement actions under factual layoff text but only two under factual financing text. Mistral remains buy-side across these treatments. The defensible cross-event conclusion is that presentation affects participation and the active mix of families; neither universal withdrawal under factual text nor a general reshuffling of directional preferences is established.

Clause order produces large descriptive differences when pooling arms: cash-first configurations generate 119 announcement actions with a buy fraction of 0.647, versus 196 and 0.413 for capacity-first configurations. Because order follows seed parity, these contrasts also differ in seed configurations. A within-seed order swap is needed for a clean order-treatment estimate. No cross-theme interaction test is reported.

\subsection{Further evidence}
Two further analyses are reported in Appendix~\ref{app:legacy}. Homogeneous populations provide a boundary case: under the original financing text, all-Mistral markets submit only buys and all-Qwen markets only sells in every announcement window. Configuration-level resampling shows that all six family-share shifts keep their sign when any single configuration is deleted, and that large share shifts need not imply large changes in net flow (Section~\ref{sec:netflow}).

\subsection{Auxiliary evidence}
Four repeated announcement batches retain buy-only active Mistral decisions and sell-only active Qwen decisions, conditional on the same announcement and reused configurations. Follow-up factual and no-update controls alter participation without uniformly altering active direction. The hold audit also finds bearish rationale text among abstaining Mistral agents, despite buy-only active Mistral decisions. The detailed tables, mandate patterns and coding limitations appear in Appendix~\ref{app:legacy}. Repeated generation, a second event, and auxiliary text evidence answer different robustness questions; none supplies a broad model- or event-population replication.

\section{A framework for active representation}\label{sec:framework}
This section organizes the evidence around three questions. When do installed shares misrepresent the families that actually supply orders? When does a change in that representation move net order flow? When does it remove opposing flow? Appendix~\ref{app:legacy} shows that pair-level agreement in Study~1 is largely accounted for by family-level directions, which is why the composition of active families, rather than coupling between individual agents, is the object we measure. Figure~\ref{fig:mechanism} summarizes the structure. All results in this section are identities or conditions on expected counts; they discipline measurement and do not model why a family participates.

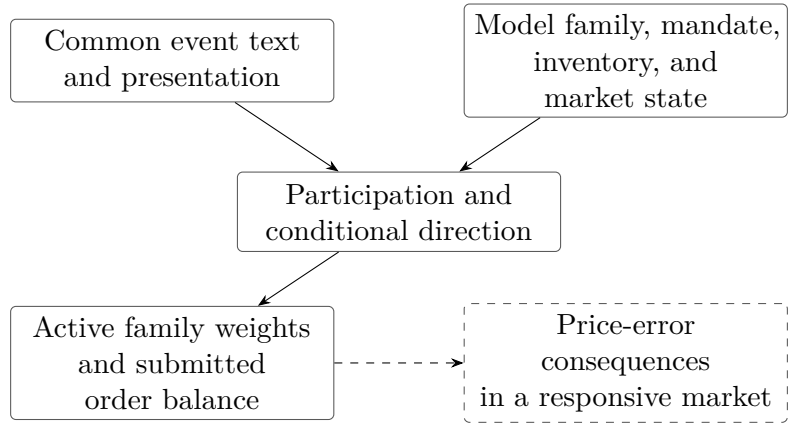
\begin{figure}[htbp]\centering\singlespacing
\begin{tikzpicture}[box/.style={draw=black!65,rounded corners=2pt,align=center,text width=4.0cm,minimum height=1.05cm,font=\small},>=Stealth]
\node[box] (text) at (0,0) {Common event text\\and presentation};
\node[box] (state) at (6,0) {Model family, mandate,\\inventory, and market state};
\node[box] (active) at (3,-2) {Participation and\\conditional direction};
\node[box] (flow) at (0,-4) {Active family weights\\and submitted order balance};
\node[box,dashed] (price) at (6,-4) {Price-error consequences\\in a responsive market};
\draw[->] (text)--(active);\draw[->] (state)--(active);
\draw[->] (active)--(flow);\draw[->,dashed] (flow)--(price);
\end{tikzpicture}
\caption{From a shared decision technology to active order flow}\label{fig:mechanism}
\par\smallskip\footnotesize The solid path organizes measured decisions and descriptive composition accounting. The dashed path requires a value benchmark and responsive liquidity; it is an analytical implication and proposed test. Arrows do not represent estimated causal mediation effects. Common market feedback remains in the simulator.
\end{figure}

\subsection{Installed and active composition}\label{sec:installed}
Fix an event, an arm $z$ and the other conditions of a cell. Let $w_f$ be family $f$'s installed share of agents and $a_f(z)$ its participation rate, the probability that an eligible agent of the family submits an order. The expected share of family $f$ among active decisions is
\begin{equation}\label{eq:activeshare}
\widetilde w_f(z)=\frac{w_fa_f(z)}{\overline a(z)},\qquad \overline a(z)=\sum_gw_ga_g(z).
\end{equation}
Let $p_f(z)$ be the family's conditional buy probability. The buy share among active decisions is then $\sum_f\widetilde w_f(z)\,p_f(z)$. Framing can therefore move aggregate direction by changing participation, even when every $p_f$ is unchanged.

\begin{proposition}[Misrepresentation and proportional response]\label{prop:installed}
\begin{enumerate}
\item[(a)] Measured by total variation, the distance between active and installed composition is
\begin{equation}\label{eq:misrep}
\tfrac12\sum_f\bigl|\widetilde w_f(z)-w_f\bigr|=\frac{\sum_fw_f\,|a_f(z)-\overline a(z)|}{2\,\overline a(z)}.
\end{equation}
Installed shares equal active shares if and only if every family with $w_f>0$ has the same participation rate.
\item[(b)] Suppose $a_f(0)>0$ for every family. The two arms have the same active composition if and only if there is a constant $k>0$ with $a_f(1)=k\,a_f(0)$ for every $f$.
\item[(c)] Equal percentage-point changes, $a_f(1)-a_f(0)=\delta$ for every $f$, leave active composition unchanged if and only if $\delta=0$ or all baseline participation rates are equal.
\end{enumerate}
\end{proposition}
\begin{proof}
For (a), equation~\eqref{eq:activeshare} gives $\widetilde w_f-w_f=w_f(a_f-\overline a)/\overline a$; summing absolute values gives equation~\eqref{eq:misrep}, which is zero exactly when every $a_f$ equals $\overline a$. For (b), if $a_f(1)=ka_f(0)$ for all $f$, then $\overline a(1)=k\,\overline a(0)$ and each ratio in equation~\eqref{eq:activeshare} is unchanged. Conversely, equal active shares require $a_f(1)/a_f(0)=\overline a(1)/\overline a(0)$ for every $f$. For (c), the ratio $a_f(1)/a_f(0)=1+\delta/a_f(0)$ is common to all families if and only if $\delta=0$ or the $a_f(0)$ are equal; part (b) completes the argument.
\end{proof}

Part (a) states when installed shares describe the families that trade. They do so only when families are equally willing to trade, and the misrepresentation equals the relative dispersion of participation rates around their mean. Part (b) states what a framing effect on composition requires. A framing that doubles every family's willingness to trade changes the number of active orders but not their composition; composition changes only when families respond in different proportions. Part (b) is the null of the Study~2 confirmatory test. With equal eligible counts per family, the estimand in equation~\eqref{eq:counttarget} is zero for every family in an event if and only if the family response ratios $a_f(1)/a_f(0)$ are equal in that event. Part (c) shows that the hypothesis of equal percentage-point changes is a different restriction: when baseline participation differs across families, equal percentage-point changes alter composition. The two hypotheses are therefore not nested, and a test of one is not a test of the other.

Study~1 shows large departures from both equal participation and proportional response. In the factual layoff arm, Mistral, Qwen and Llama participate at the announcement in 14.0\%, 48.0\% and 18.0\% of opportunities; in the explanatory arm, the rates are 74.7\%, 40.0\% and 15.3\%. Explanatory text multiplies Mistral's participation by 5.3, while Qwen's and Llama's fall by 17\% and 15\%. In the financing event, the original text multiplies participation relative to factual text by 3.3 for Mistral, 33 for Qwen and 0.44 for Llama. Although each family is one-third of the installed population in every arm, the total-variation distance between installed and active composition is 0.27 and 0.24 in the factual and explanatory layoff arms and 0.18 and 0.45 in the original and factual financing arms. In these markets, an observer who counted installed models would misstate the composition of active orders by 18 to 45 percentage points.

\subsection{When changes in representation move net flow}\label{sec:netflow}
Proposition~\ref{prop:installed} concerns whether representation changes. This subsection concerns whether a change in representation matters for net order flow. Both presentation contrasts move some family's active share by 40 to 61 percentage points, but only the layoff contrast moves net submitted orders. The accounting below separates the two cases and applies to any number of families.

For arm $z$, let $n_f(z)$ be family $f$'s active decisions and $Q_f(z)$ its buy decisions minus its sell decisions. Net submitted orders are $Q(z)=\sum_fQ_f(z)$, total activity is $N(z)=\sum_fn_f(z)$, and active shares are $\widetilde w_f(z)=n_f(z)/N(z)$. When $n_f(z)>0$, the family's conditional direction is $d_f(z)=Q_f(z)/n_f(z)=2p_f(z)-1\in[-1,1]$, and the mean direction of active decisions is $m(z)=Q(z)/N(z)=\sum_f\widetilde w_f(z)d_f(z)$. Changes are arm~1 minus arm~0, and bars denote arm midpoints.

\begin{proposition}[Reinforcing and offsetting selection]\label{prop:netflow}
\begin{enumerate}
\item[(a)] $|\Delta Q|=R\sum_f|\Delta Q_f|$ for a ratio $R\in[0,1]$, and $R=1$ if and only if all nonzero family contributions $\Delta Q_f$ have the same sign.
\item[(b)] If each family's conditional direction is the same in both arms, then $\Delta Q_f=d_f\,\Delta n_f$. In that case $R=1$ if and only if the presentation change raises participation only among families with one sign of $d_f$ and lowers it only among families with the opposite sign. A change that raises participation among both net-buying and net-selling families produces offsetting contributions.
\item[(c)] When every family-arm direction is defined,
\begin{equation}\label{eq:netflow}
\Delta Q=\underbrace{\Delta N\,\overline m}_{\text{scale}}
+\underbrace{\overline N\cdot TV\cdot(\delta^+-\delta^-)}_{\text{composition}}
+\underbrace{\overline N\sum_f\overline w_f\,\Delta d_f}_{\text{conditional direction}},
\end{equation}
where $TV=\tfrac12\sum_f|\Delta\widetilde w_f|$ is the share of active decisions reallocated across families, and $\delta^+$ and $\delta^-$ are the averages of $\overline d_f$ over families that gain and lose active share, weighted by $|\Delta\widetilde w_f|$. The composition term is zero when $TV=0$, and its absolute value is at most $\overline N\cdot TV\cdot(\max_f\overline d_f-\min_f\overline d_f)$.
\end{enumerate}
\end{proposition}
\begin{proof}
Part (a) is the triangle inequality, which holds with equality exactly when the nonzero terms share a sign. For part (b), $Q_f=n_fd_f$ with $d_f$ fixed gives $\Delta Q_f=d_f\Delta n_f$, whose sign is the product of the signs of $d_f$ and $\Delta n_f$. For part (c), any product of two factors satisfies $x_1y_1-x_0y_0=\Delta x\,\overline y+\overline x\,\Delta y$. Applying this to $Q=Nm$ and then to each term of $m=\sum_f\widetilde w_fd_f$ gives the three terms, with composition term $\overline N\sum_f\Delta\widetilde w_f\overline d_f$. Because $\sum_f\Delta\widetilde w_f=0$, the positive share changes sum to $TV$ and the negative changes to $-TV$, so $\sum_f\Delta\widetilde w_f\overline d_f=TV(\delta^+-\delta^-)$. The bound holds because $\delta^+$ and $\delta^-$ are weighted averages of the $\overline d_f$.
\end{proof}

Equation~\eqref{eq:netflow} extends the buy-share accounting in equation~\eqref{eq:accounting}. Since $m=2B-1$, its last two terms equal $2\overline N$ times the active-composition and conditional-direction components of that equation, and the scale term adds the change in the number of active decisions. Part (c) shows what makes a reallocation consequential for net flow. A shift in representation among families with the same average direction leaves net flow unchanged however large $TV$ is. A shift from the most sell-oriented family toward the most buy-oriented family has the largest effect. Parts (a) and (b) require no convention for family-arm cells with few or no active decisions.

\begin{table}[htbp]\centering\small\singlespacing
\caption{Net submitted orders: family contributions and accounting components}\label{tab:flowaccounting}
\begin{tabular}{lrr}\toprule
 & Financing & Layoff\\
 & original $-$ factual & explanatory $-$ factual\\\midrule
\multicolumn{3}{l}{\emph{A. Family contributions, $\Delta Q_f$}}\\
Mistral & 28 & 91\\
Qwen & $-68$ & 20\\
Llama & 37 & $-2$\\\addlinespace
Change in net orders, $\Delta Q$ & $-3$ & 109\\
Sum of absolute contributions, $\sum_f|\Delta Q_f|$ & 133 & 113\\
Reinforcement ratio, $R$ & 0.02 & 0.96\\\midrule
\multicolumn{3}{l}{\emph{B. Components in equation~\eqref{eq:netflow}}}\\
Active-share reallocation, $TV$ & 0.61 & 0.40\\
Scale & $-19.5$ & $-6.3$\\
Composition & 39.2 & 108.4\\
Conditional direction & $-22.8$ & 6.9\\\addlinespace
Composition at either arm's directions & $[-26.7,\ 105.2]$ & $[106.1,\ 110.7]$\\\bottomrule
\end{tabular}
\par\smallskip\footnotesize Announcement round $t=5$, pooled over 30 paired configurations. $Q_f$ is family $f$'s buy decisions minus sell decisions, excluding background liquidity; these are submitted decisions, not executions or share volume. Panel B evaluates conditional directions at arm midpoints; the last row replaces the midpoint directions in the composition term with those of each arm in turn. Qwen's factual-arm direction in the financing event rests on two active decisions.
\end{table}

Table~\ref{tab:flowaccounting} applies the proposition to the two events. In the layoff contrast the family contributions reinforce one another. Explanatory text adds 91 net Mistral buys, Qwen sells 20 fewer on net, and Llama's net orders change by $-2$. Net orders change by 109, so $R=0.96$. Conditional directions are nearly stable across arms: every active Mistral decision is a buy in both arms, and Qwen's direction is $-0.70$ under explanatory and $-0.86$ under factual text. Forty percent of active decisions are reallocated, from families whose average midpoint direction is $-0.72$ to Mistral, whose direction is $+1$. The composition term accounts for 108 of the 109-order change and is close to its upper bound of 112. Evaluating it at either arm's directions instead of the midpoint gives 106 to 111.

In the financing contrast the family contributions offset. Original text adds 68 net Qwen sells, while Mistral adds 28 net buys and Llama sells 37 fewer on net. The absolute contributions sum to 133, but net orders change by three, so $R=0.02$. Gross submitted orders nearly double, from 66 under factual text to 129 under original text, so the original text adds orders on both sides of the market while leaving net flow almost unchanged. Original text draws a buy-side family and a sell-side family into trading at the same time, and Llama, which sells under factual text, trades less. The split between composition and conditional direction is not stable in this event. Qwen moves from 66 sells under original text to two buys under factual text, so the composition term ranges from $-27$ to 105 depending on which arm's directions are used. The family-level offset in part~(a) is the robust description.

The distinction matters for prices. Under the responsive supply rule specified for Study~2, the midpoint displacement is proportional to executed net flow (equation~\eqref{eq:displace}). Reinforcing selection therefore moves price pressure, while offsetting selection leaves it unchanged. Offsetting selection is still not economically neutral. In the market-clearing benchmark of Section~\ref{sec:pricing}, the price error depends on the demand-weighted valuation errors of the families that participate, so a reallocation that leaves net flow unchanged can still change which errors enter the price (Proposition~\ref{prop:selection}). In the financing contrast, a measure based on net order flow would show almost no response, while the families supplying demand change substantially. The proposition is an accounting result: it does not explain why families participate, but it states which feature of a participation response determines its effect on net flow.

\subsection{Composition and opposing flow}
Composition also determines whether submitted orders meet orders on the other side of the market, and pair-level agreement in Study~1 is largely accounted for by family-level directions rather than by residual dependence between agents. Appendix~\ref{app:legacy} gives the one-sidedness benchmark, the identities linking pair agreement to family marginals, and the permutation diagnostics.

\section{Active weights and economic accuracy: analytical benchmarks}\label{sec:pricing}
The current simulator constrains prices with thick pegged liquidity. To state what would be needed for an economic consequence, we add a transparent linear-demand benchmark. It is not estimated on the logs, is not a rational-expectations solution, and does not assume that disagreement itself is an error. It illustrates a distinction between individual precision and the covariance of deployed readings.

The possibility that algorithmic trading weakens price efficiency is not itself new. \citet{dou} studies reinforcement-learning trading and algorithmic collusion. Here the analytical channel is the covariance of deployed interpretation errors, not learned strategic restraint. Study 2 below does not estimate this price-error model: it has no independent fundamental value.

Let the terminal fundamental be $v$. An outside sector has a valuation $v+e_0$ and aggregate net supply $\lambda[P-(v+e_0)]$, where $\lambda>0$ is its price sensitivity. This does not give that sector knowledge of $v$: $e_0$ is its valuation error. Model family $f$ has valuation $v+u_f$ and aggregate demand $K_f(v+u_f-P)$, where $K_f\geq0$. The common family error abstracts from residual individual errors. Net noise demand is $\nu$. Clearing gives
\begin{equation}
\lambda[P-(v+e_0)]=\sum_f K_f(v+u_f-P)+\nu.
\end{equation}
Let $K=\sum_fK_f>0$ and $\omega_f=K_f/K$. Then
\begin{equation}\label{eq:price}
P-v=\frac{\lambda e_0+K\sum_f\omega_fu_f+\nu}{\lambda+K}.
\end{equation}
The weights $\omega_f$ are demand-slope shares. The next result gives the condition under which they coincide with active shares, which connects the benchmark to the composition measured in the experiments.

\begin{proposition}[Participation and price-error weights]\label{prop:weights}
Suppose each active agent of family $f$ has demand $\kappa(v+u_f-P)$ with a common $\kappa>0$, inactive agents submit no demand, and $N$ agents are active. Then $\omega_f=\widetilde w_f$, $K=\kappa N$, and
\begin{equation}\label{eq:weights}
P-v=\frac{\lambda e_0+\nu}{\lambda+\kappa N}+\alpha(N)\sum_f\widetilde w_fu_f,\qquad \alpha(N)=\frac{\kappa N}{\lambda+\kappa N}.
\end{equation}
Framing therefore affects the price error through two channels: the number of active agents, which sets the weight $\alpha(N)$ on model errors relative to the outside sector, and active composition, which sets which model errors receive that weight. If participation responds in the same proportion in every family, only the first channel operates.
\end{proposition}
\begin{proof}
Family $f$ has $n_f=N\widetilde w_f$ active agents, so $K_f=\kappa N\widetilde w_f$, $K=\kappa N$ and $\omega_f=\widetilde w_f$. Substituting into equation~\eqref{eq:price} gives equation~\eqref{eq:weights}. The last statement follows from Proposition~\ref{prop:installed}(b).
\end{proof}

The two channels parallel the scale and composition terms in equation~\eqref{eq:netflow}. The common sensitivity $\kappa$ is a benchmark assumption: the experiments cap orders at five shares, and we do not impose $\omega_f=\widetilde w_f$ on the capped-order data. Within the benchmark, Proposition~\ref{prop:weights} also states why the Study~2 estimand matters economically. Its null is that framing leaves unchanged the weights on which the price error loads, while allowing the total weight on model errors to change.

\begin{proposition}[Precision, active concentration, and price error]\label{prop:price}
Suppose $e_0$, $\nu$, and the vector of family errors have zero means and finite variances; $e_0$ and $\nu$ are mutually uncorrelated and uncorrelated with that vector. Suppose each $u_f$ has variance $\sigma_u^2$, with common pairwise correlation $\rho$ forming a positive semidefinite correlation matrix. Holding the weights and demand sensitivities fixed independently of realized errors, define $H=\sum_f\omega_f^2$. Then
\begin{equation}\label{eq:mse}
\E[(P-v)^2]=\frac{\lambda^2\sigma_0^2+\sigma_\nu^2+K^2\sigma_u^2[\rho+(1-\rho)H]}{(\lambda+K)^2}.
\end{equation}
For $\sigma_u^2>0$ and $\rho<1$, this mean squared price error increases strictly with $H$ at fixed $K$. At fixed $H$, define $\tau^2=\sigma_u^2[\rho+(1-\rho)H]>0$ and $A=\lambda^2\sigma_0^2+\sigma_\nu^2>0$. The error decreases with $K$ for $K<A/(\lambda\tau^2)$ and increases for $K>A/(\lambda\tau^2)$.
\end{proposition}
\begin{proof}
The variance of the weighted family error is
$\sigma_u^2\sum_f\omega_f^2+2\rho\sigma_u^2\sum_{f<g}\omega_f\omega_g
=\sigma_u^2[\rho+(1-\rho)H]$.
Applying the stated orthogonality assumptions to equation~\eqref{eq:price} yields equation~\eqref{eq:mse}. Its derivative with respect to $H$ is $K^2\sigma_u^2(1-\rho)/(\lambda+K)^2>0$. Differentiating $(A+K^2\tau^2)/(\lambda+K)^2$ with respect to $K$ gives $2(K\lambda\tau^2-A)/(\lambda+K)^3$, which establishes the threshold.
\end{proof}

The result makes the connection to predictive technologies precise without overturning it. Lower individual error variance improves the benchmark at fixed weights and adoption; independent readings can further reduce aggregate error. But stronger aggregate demand places more weight on the deployed family's error, and correlated errors limit the gain from replication. Improvement is not determined by the number of model instances alone. This is an illustrative diversification calculation within a market-clearing equation, not a claim of novelty for the underlying variance identity.

For a concrete uncalibrated example, set $\lambda=1$, $K=3$, $\sigma_0^2=1$, $\sigma_\nu^2=0$, and $\rho=0$. Three equally represented families with $\sigma_u^2=1$ produce mean squared price error $(1+9/3)/16=0.25$. A single family with improved individual variance $\sigma_u^2=0.5$ produces $(1+9/2)/16=0.34375$. Each model can thus be more accurate while the resulting concentrated configuration has greater price error under these assumptions. These numbers are an analytical illustration; no experimental price-error estimate is attached to them.

\begin{corollary}[Precision required to offset concentration]
Under Proposition~\ref{prop:price}, compare a diversified configuration with family variance $\sigma_D^2$ and concentration $H_D$ with a single-family configuration of variance $\sigma_C^2$, holding $K$, $\lambda$, and outside-sector and noise risks fixed. The concentrated configuration has weakly lower price-error variance if and only if
\begin{equation}
\sigma_C^2\leq\sigma_D^2[\rho+(1-\rho)H_D].
\end{equation}
\end{corollary}
\begin{proof}
The common denominator and outside-risk terms in equation~\eqref{eq:mse} cancel. A single family has $H=1$, so comparing the remaining weighted-error variances yields the inequality.
\end{proof}
For three equally weighted, uncorrelated families, the variance of a single family must be at most one-third of the diversified families' individual variance to offset the loss of diversification. This condition concerns forecast-error variance under the stated demand assumptions; a ceiling-level accuracy score from an unambiguous task cannot establish it. Figure~\ref{fig:precision} illustrates this trade-off and the adoption threshold.

\begin{figure}[htbp]\centering
\includegraphics[width=\linewidth]{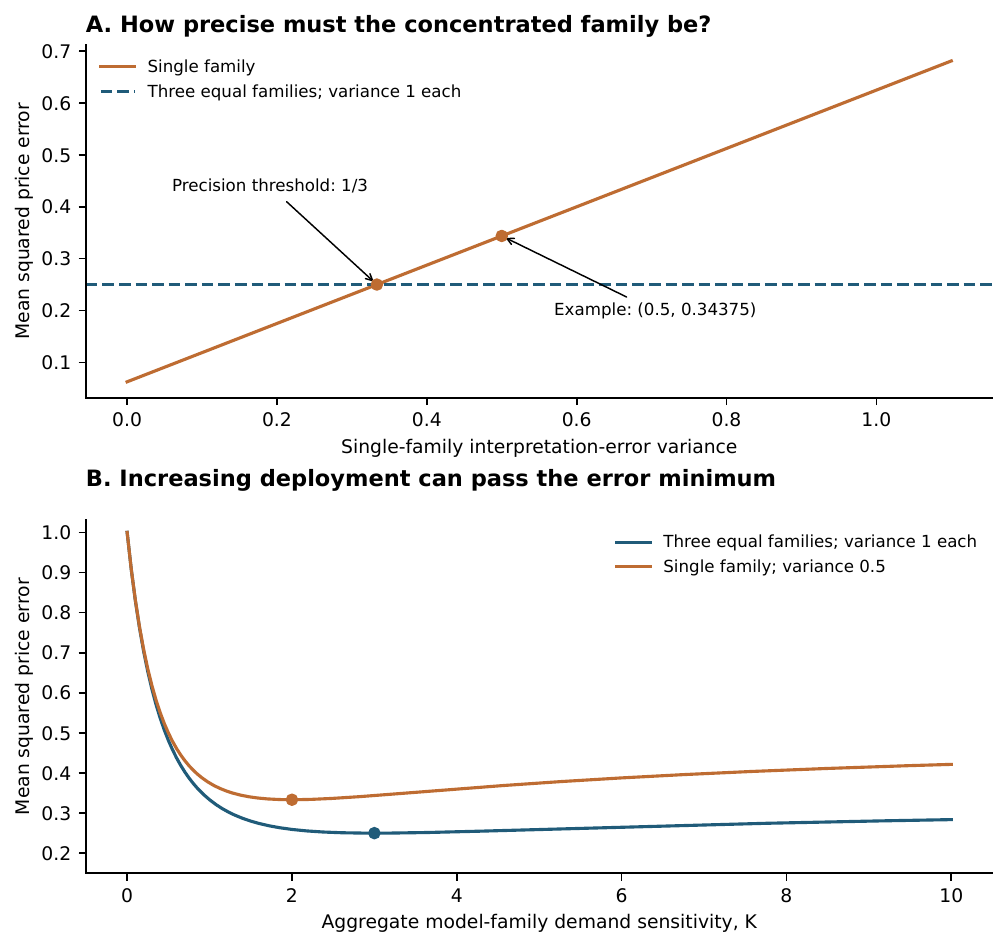}
\caption{Individual precision, active concentration, and price error}
\label{fig:precision}
\par\smallskip\footnotesize Analytical illustration, not an estimate. Panel A fixes $K=3$, $\lambda=1$, $\sigma_0^2=1$, $\sigma_\nu^2=0$, and $\rho=0$. The horizontal reference is three equal families, each with error variance one. Panel B varies $K$ with the same outside risks and compares three equal families of variance one to a single family of variance one-half. Markers identify the minima given by Proposition~\ref{prop:price}.
\end{figure}

\paragraph{Forecast dispersion is not error correlation.}
The comparison with \citet{bali} requires distinguishing both the outcome and the source of disagreement. That paper links machine forecast disagreement to the cross-section of subsequent returns, with evidence consistent with short-sale costs and limits to arbitrage. Proposition~\ref{prop:price} concerns mean squared error of a clearing price under unrestricted linear demands. In our benchmark,
\begin{equation}
\frac{\partial\E[(P-v)^2]}{\partial\rho}
=\frac{K^2\sigma_u^2(1-H)}{(\lambda+K)^2}.
\end{equation}
Thus a lower $\rho$ reduces price-error variance for $H<1$ and $\sigma_u^2>0$, holding other parameters fixed; at $H=1$ there is no cross-family covariance term to diversify. This is not a claim that any increase in forecast disagreement improves markets.

To see the distinction, define weighted cross-family dispersion as
$V_\omega=\sum_f\omega_f(u_f-\sum_g\omega_gu_g)^2$. Under the same zero-mean, equal-variance, fixed-weight assumptions,
\begin{equation}
\E[V_\omega]=\sigma_u^2(1-\rho)(1-H).
\end{equation}
Indeed, weighted dispersion equals $\sum_f\omega_fu_f^2-(\sum_f\omega_fu_f)^2$; taking expectations gives the result. Lower correlation raises dispersion and lowers aggregate error at fixed marginal variance. Raising marginal variance can instead raise both dispersion and price error. Consequently, the empirical return result and this conditional diversification calculation do not contradict one another. The quality battery must measure error levels, bias and cross-family covariance separately; a scalar disagreement measure cannot supply them.

The boundaries of the proposition are material. With unequal variances or an arbitrary covariance matrix, the relevant expression is $\omega'\Sigma_u\omega$, and concentration alone does not rank configurations. With shared bias, the squared weighted mean error must also be included. With endogenous selection into participation, weights may covary with errors, invalidating the fixed-weight calculation. With equilibrium learning, outside-sector behavior and error covariance may change with adoption. Our evidence documents precisely why the active weights may respond to text; it does not verify the benchmark's covariance assumptions.

The proposition yields three conditional tests for a new market experiment. First, at matched individual forecasting quality, changing active error covariance should predict price-error differences more reliably than counting provider names. Second, text changes that increase participation in an erroneous reading should have larger price effects when outside supply is less elastic, holding total demand sensitivity fixed. Third, adding model instances within a family need not replicate the benefit of adding independent information. Testing these predictions requires a fundamental-value benchmark and independently measured errors; directional unanimity in the present logs is insufficient.

\subsection{Selection into active demand}\label{sec:endogenous}
The fixed-weight calculation does not cover the mechanism documented in the experiments: weights can respond to the same readings that generate valuation errors. Let $w^0$ be fixed installed demand shares and $W$ be random active demand shares, each summing to one. Keep aggregate demand sensitivity $K$ fixed to isolate reallocation; the empirical experiments need not keep it fixed. Define $Z=(w^0)'u$, $S=(W-w^0)'u$, and $\alpha=K/(\lambda+K)$. Then the active weighted error is $Z+S$.

\begin{proposition}[Selection term in price accuracy]\label{prop:selection}
Suppose the outside/noise term $L=(\lambda e_0+\nu)/(\lambda+K)$ is mean zero and independent of $(u,W)$, with all displayed second moments finite. Comparing the installed-weight benchmark with active weights at the same $K$,
\begin{equation}\label{eq:selectionrisk}
\E[(P_W-v)^2]-\E[(P_0-v)^2]
=\alpha^2\{\E[S^2]+2\E[ZS]\}.
\end{equation}
Even if every family error is unconditionally unbiased, active aggregate bias can equal $\alpha\E[S]\ne0$. Selection improves mean squared accuracy if and only if $\E[ZS]<-\E[S^2]/2$ (with equality giving no change).
\end{proposition}
\begin{proof}
Equation~\eqref{eq:price} gives $P_W-v=L+\alpha(Z+S)$ and $P_0-v=L+\alpha Z$. Subtracting their squared expectations leaves $2\alpha\E[LS]+\alpha^2\E[2ZS+S^2]$. Independence and the zero mean of $L$ remove the first term. Unbiased family errors imply $\E[Z]=0$, which gives the bias claim and the inequality.
\end{proof}

This is an exact risk decomposition, not a behavioral law of selection. It gives two distinct possibilities: selection can overweight errors or select against them. Standalone accuracy and an active Herfindahl are insufficient to determine the sign. To see the latter sharply, consider two equally installed families. With equal probabilities their error vector is $(a,0),(-a,0),(0,a),(0,-a)$, for $a>0$. Both have mean zero and variance $a^2/2$, and their covariance is zero. Equal weights give aggregate squared error $a^2/4$. Selecting only the family with the nonzero error gives $a^2$; selecting only the zero-error family gives zero. Both selections have active concentration one, the same marginal quality distribution and the same $K$. These are constructed information-dependent selection rules, not estimated model policies. They demonstrate why a concentration statistic cannot rank accuracy without a model of selection.

An empirical test requires values or errors for nonparticipants as well as participants, measured against a common benchmark. Our existing order logs do not contain those errors, and generated rationales cannot substitute for them. The known-value validation below supplies the missing outcome and a separate elicitation design. A fixed-weight quality comparison remains useful, but cannot validate equation~\eqref{eq:selectionrisk} by itself. If outside-sector errors respond to active selection, the omitted cross term $2\alpha\E[LS]$ must be retained; the independence restriction is substantive.

\section{Prospective test of evaluative framing}\label{sec:study2}
Study~2 asks whether explicit evaluative framing changes which model families supply active orders when the installed mixture and the facts of each announcement are held fixed. It answers the question in six new corporate events, with a preregistered test, 18,000 runs in the main sample and a separately worded replication of the same size. Framing changes active composition in both samples. Adding evaluations to the same facts also removes most of the net buying that the factual text produces, in every event of both samples, and the family responses mostly reinforce rather than offset one another in net flow. The composition shifts themselves recur across the two wordings in three of the six events. Appendix~\ref{app:study2} gives the protocol, the development of the test and the record of deviations.

\paragraph{Design.}
The six themes are a customer contract, an inventory purchase, a patent licence, an operating permit, a division sale and a product redesign. Each theme has a fixed factual inventory of three sentences, and both arms contain exactly those sentences. The evaluative arm adds two sentences: one calls a stated fact favorable and the other calls a stated fact unfavorable. It adds no new facts and no statements about other traders' beliefs. Each text exists in two clause orders, and one order is drawn for each configuration pair and shared by its two arms. Because the evaluative sentences also lengthen the text and the design has no length-matched control, the treatment effect is the effect of adding explicit evaluations to a fixed set of facts; it does not separate evaluation from added length. The market is that of Study~1 with a responsive liquidity provider whose midpoint moves by $1/\lambda$ per share of net agent purchases (equation~\eqref{eq:displace}), with $\lambda=20$.

Four reviewers assessed every main text independently against its inventory, and all 24 texts were admitted without wording changes before the pilot. One of the authors wrote a second, independently worded realization of each theme from the same inventories, archived it before the pilot, and submitted it to the same review. Three of the four reviewers completed that review; the fourth withdrew for lack of time before any review was compared, and the admission rule was applied unchanged to the three reviews. All 24 replication texts were admitted without changes. Appendix~\ref{app:stimuli} reproduces both sets. The review establishes properties of the texts, not the behavioral mechanism.

\paragraph{Confirmatory test.}
H2 is the sole primary confirmatory hypothesis. Its null is that, in every event, each family's expected active count at the announcement round is the same share of the expected total active count in both arms (equation~\eqref{eq:counttarget}). By Proposition~\ref{prop:installed}(b), the null holds in an event exactly when framing scales every family's participation by the same factor, so a rejection is evidence of non-proportional participation responses in at least one event. The test is a bounded betting test \citep{betting}. For each event, family and fixed stake it forms a product e-value from bounded scores on disjoint blocks of configuration pairs, averages these e-values over the six events, three families and 18 stakes, and rejects at the 5\% level when the average $E_{H2}$ is at least 20. An average of e-values is an e-value, so no multiplicity correction is needed, and $p=\min(1,1/E_{H2})$ is a valid p-value. Validity requires independent and identically generated configuration pairs within each event but no model of the counts. The test became primary after three other procedures failed prespecified checks in simulation. The replication is analysed with the same test, applied unchanged and reported separately; the two e-values are never combined, and the primary claim rests on the main sample. Composition vectors, family participation, the flow accounting and the H3 decomposition are prespecified as supporting, exploratory or descriptive analyses and carry no tests.

\paragraph{Samples.}
A power study under nine alternatives, each with a total-variation distance of .10 in every affected event, selected 1,500 configuration pairs per event as the smallest tested size at which a simultaneous lower confidence bound for power reached .80 in every alternative (Table~\ref{tab:power}). A pilot of 72 twelve-round runs selected an eight-round horizon. Because H2 is measured at the announcement round, each campaign generates the complete eight-round schedule and executes rounds 0 to 5. Prompts contain neither the round number nor the horizon, so the executed rounds are generated by the same process as the first six rounds of a full-horizon run. The preregistration was frozen on 28 October 2026. The main campaign ran from 28 October to 11 November and the replication from 14 to 16 November 2026, on NVIDIA RTX A6000 GPUs of a university computing cluster. Every scheduled run was completed, every interrupted run was re-executed under the prespecified failure rule, and all technical checks passed in both campaigns (Table~\ref{tab:s2_h2}). Outcomes were opened only after the replication texts had been admitted, the fixed sample was complete and the technical checks had passed.

\begin{table}[htbp]\centering\small\singlespacing
\caption{Completed samples and the preregistered H2 test, by campaign}\label{tab:s2_h2}
\resizebox{\linewidth}{!}{%
\begin{tabular}{lrrrrrrlrl}\toprule
Campaign & Runs & Pairs & Calls & \shortstack{Runs with\\one failure} & \shortstack{Max\\failures} & \shortstack{Parse\\failures} & $E_{H2}$ & $\log E_{H2}$ & Decision\\
\midrule
Main & 18,000 & 9,000 & 1,620,000 & 47 & 1 & 0 & $6.301\times10^{25}$ & 59.41 & reject\\
Replication & 18,000 & 9,000 & 1,620,000 & 19 & 1 & 0 & $6.662\times10^{40}$ & 94.00 & reject\\
\bottomrule
\end{tabular}}
\par\smallskip\footnotesize All scheduled runs committed in both campaigns; every failure was rule-covered and re-executed, at most one per run. Decision rule $E_{H2}\ge 20$ (5\% level); valid $p\le 1/E_{H2}$. The two campaigns are tested separately and never combined; the primary claim rests on the main campaign. A larger e-value is not a larger effect. Per-event log e-values, descriptive only (events in the order of Table~\ref{tab:s2_composition}): main 34.5, 61.2, 26.9, 23.1, 46.8, 20.4; replication 54.4, 2.7, 14.6, 0.1, 24.2, 95.8.
\end{table}

\paragraph{Framing changes active composition.}
Both tests reject. The main sample gives $E_{H2}=6.3\times10^{25}$ and the replication $6.7\times10^{40}$, against the threshold of 20 (Table~\ref{tab:s2_h2}). No reply in either sample failed to parse. Within the six themes and their two wordings, evaluative framing therefore changes participation in different proportions across families, which is the condition under which installed and active composition come apart. The test is global and does not identify which events depart from the null. The per-event e-values, reported descriptively, show that each sample's average is dominated by a single event, inventory purchase in the main sample and product redesign in the replication. The larger replication e-value reflects that event and does not indicate a larger effect.

\begin{figure}[htbp]\centering
\includegraphics[width=\linewidth]{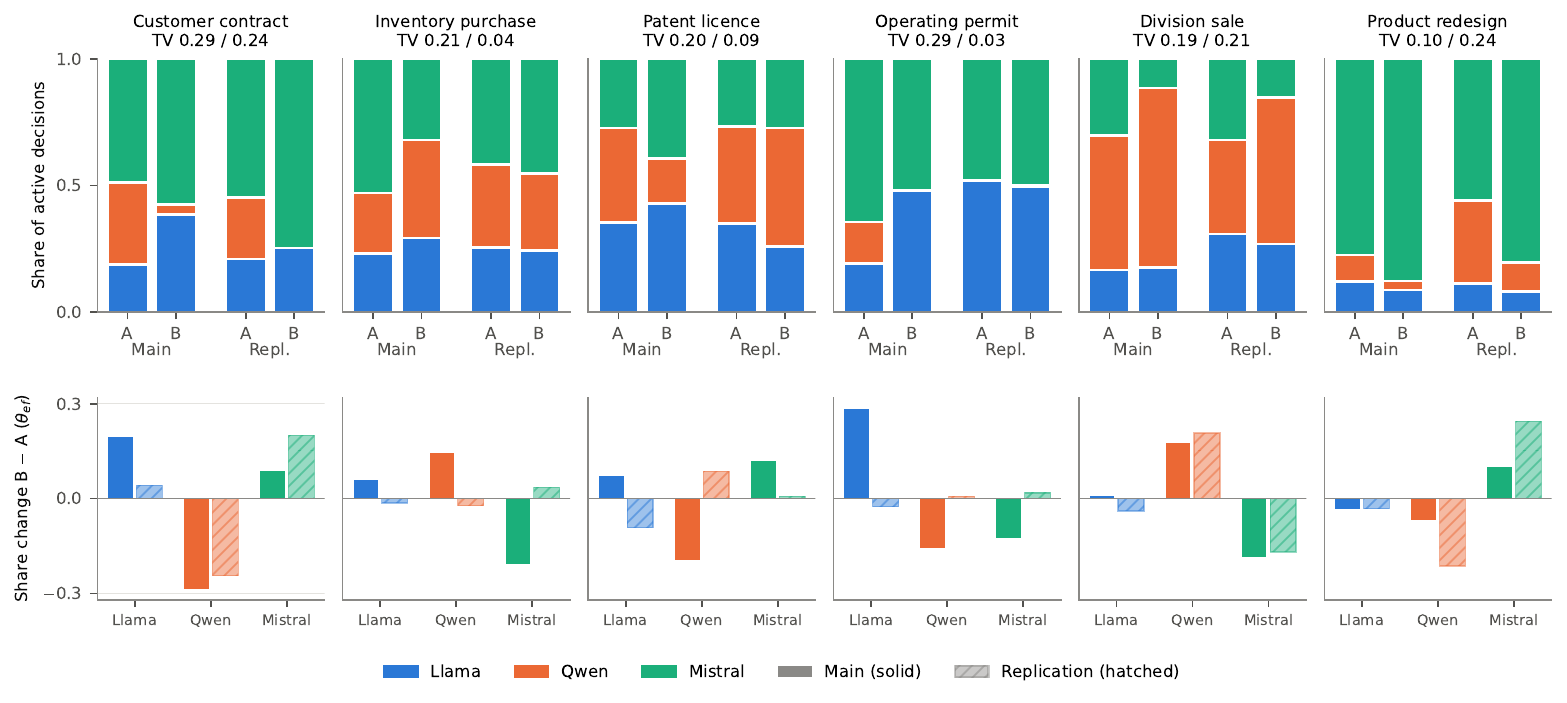}
\caption{Active composition at the announcement round, main sample and wording replication. Top: each family's share of active decisions in the factual (A) and evaluative (B) arm; in each panel the left pair is the main sample and the right pair the replication. Bottom: share change $\theta_{ef}$ (B $-$ A), main sample solid and replication hatched. Panel titles give the within-event total-variation distance, main / replication. 1,500 configuration pairs per event and sample.}\label{fig:s2composition}
\end{figure}

Figure~\ref{fig:s2composition} shows the composition shifts behind the rejections, and Table~\ref{tab:s2_composition} in Appendix~\ref{app:study2} gives the values. The average within-event total-variation distance is .21 in the main sample and .14 in the replication: on average, evaluative framing reallocates a fifth and a seventh of active decisions across families. Pooling the six events would hide most of this, because shifts in opposite directions cancel; the composition pooled over events moves by only .06 in the main sample. Three events show the same reallocation in both wordings. In the customer contract, Qwen's share falls by 29 and 24 points. In the division sale, Qwen gains 18 and 21 points while Mistral loses 19 and 17. In the product redesign, Mistral gains 10 and 24 points while Qwen loses 7 and 21. In the other three events the shift does not recur. Inventory purchase and operating permit reallocate .21 and .29 of active decisions in the main sample but only .04 and .03 in the replication, and in the patent licence the share changes of Llama and Qwen reverse sign. Across the 18 event-family cells, the correlation between the two samples' share changes is .43. Evaluative framing thus changes which families trade in both wordings, but which families gain depends on the event and its wording.

\paragraph{All three families trade less.}
Averaged over events, every family participates less under evaluative text (Table~\ref{tab:s2_participation}): the share of a family's five agents that submit an order falls by 9, 17 and 33 percentage points for Llama, Qwen and Mistral in the main sample, from .23, .32 and .57, and by 11, 11 and 16 points in the replication. The pair-level standard errors of these averages are about .003. Proportionally the responses differ: in the main sample Mistral's participation falls by 58\%, Qwen's by 52\% and Llama's by 39\%. In two main-sample events Qwen almost stops trading under evaluative text, with participation falling from .44 to .01 in the customer contract and from .16 to .001 in the operating permit. Total active decisions per market fall from 5.6 to 2.7 in the main sample and from 4.9 to 3.0 in the replication. They fall in all six main-sample events and in four of the six replication events; in the replication patent licence and division sale they rise by 0.1 and 0.5.

\begin{table}[htbp]\centering\small\singlespacing
\caption{Submitted orders, executed quantities, and gross and net flow at the announcement round (mean per market, equal-weight average of the six events)}\label{tab:s2_flow}
\resizebox{\linewidth}{!}{%
\begin{tabular}{lrrrrrrrr}\toprule
& \multicolumn{4}{c}{Main} & \multicolumn{4}{c}{Replication}\\\cmidrule(lr){2-5}\cmidrule(lr){6-9}
& A & B & B$-$A (SE) & $<0$ & A & B & B$-$A (SE) & $<0$\\
\midrule
Buy orders submitted & 4.58 & 1.31 & $-$3.27\,(.018) & 6/6 & 3.52 & 1.52 & $-$2.00\,(.021) & 6/6\\
Sell orders submitted & 1.05 & 1.38 & 0.33\,(.014) & 2/6 & 1.36 & 1.48 & 0.12\,(.017) & 4/6\\
Net submitted orders (buy $-$ sell) & 3.53 & $-$0.07 & $-$3.60\,(.023) & 6/6 & 2.16 & 0.04 & $-$2.12\,(.029) & 6/6\\
Gross submitted shares & 15.42 & 7.04 & $-$8.38\,(.068) & 6/6 & 12.88 & 7.87 & $-$5.00\,(.073) & 4/6\\
Net submitted shares & 10.53 & 0.76 & $-$9.77\,(.069) & 6/6 & 6.19 & 0.99 & $-$5.20\,(.082) & 6/6\\
Executed with provider, gross shares & 8.02 & 3.43 & $-$4.58\,(.053) & 5/6 & 7.70 & 4.11 & $-$3.59\,(.055) & 5/6\\
Executed with provider, net shares ($Q$) & 6.09 & $-$0.18 & $-$6.26\,(.058) & 6/6 & 3.79 & 0.13 & $-$3.66\,(.064) & 6/6\\
Executed against agents, gross (diagnostic) & 1.95 & 0.80 & $-$1.14\,(.036) & 6/6 & 1.56 & 0.85 & $-$0.71\,(.034) & 6/6\\
Midpoint displacement $Y=Q/\lambda$ & 0.30 & $-$0.01 & $-$0.31\,(.003) & 6/6 & 0.19 & 0.01 & $-$0.18\,(.003) & 6/6\\
\bottomrule
\end{tabular}}
\par\smallskip\footnotesize Submitted orders are cleaned buy or sell decisions with positive quantity at $t=5$ (parse failures are holds). Executed quantities are shares of $t=5$ orders that traded with the responsive provider in round 5 ($Q$ = shares bought from minus sold to the provider), reconciled with the provider's state change in every run; executions against other agents are a diagnostic and not part of $Q$. $\lambda=20$. $Y$ is a downstream consequence of the supply rule, not evidence about price discovery or efficiency. Column $<0$: events with B $-$ A below zero. Per-event values are in the results package. Descriptive; no test.
\end{table}

\paragraph{Evaluative text removes net buying pressure.}
Table~\ref{tab:s2_flow} separates submitted orders, executed quantities, and gross and net flow. Under factual text, active agents are on balance net buyers at the announcement in most events. Adding one favorable and one unfavorable evaluation cuts buy orders per market from 4.6 to 1.3 in the main sample and from 3.5 to 1.5 in the replication, while sell orders rise slightly, from 1.1 to 1.4 and from 1.4 to 1.5. Net submitted orders fall from 3.5 to $-0.1$ and from 2.2 to 0.0. Executed net flow with the provider falls from 6.1 to $-0.2$ shares per market in the main sample and from 3.8 to 0.1 in the replication, and it is lower under evaluative text in all twelve event-sample cells. Gross executed quantity roughly halves. Under the provider's supply rule, the announcement-round midpoint displacement falls by 0.31 and 0.18 price units. This displacement is a consequence of the supply rule, not a measure of price discovery or accuracy.

\paragraph{Family responses mostly reinforce one another.}
Proposition~\ref{prop:netflow} distinguishes reinforcing from offsetting selection, and the contributions in Study~2 are predominantly reinforcing. Because the families reduce their net buying together, their contributions to the change in executed net flow usually share a sign. The reinforcement ratio $R=|\Delta Q|/\sum_f|\Delta Q_f|$ is at least .89 in every main-sample event and at least .76 in five of the six replication events (Table~\ref{tab:s2_h3}). Reinforcement is the typical pattern rather than a property of every family in every event: a ratio below one means that at least one family moves net flow against the total. Mistral makes the largest negative contribution in all six main-sample events and in four replication events. The exception is the replication patent licence, where Qwen's lower net buying is partly offset by small increases for Llama and Mistral and $R=.47$. The financing contrast of Study~1 showed that a large reallocation can leave net flow unchanged; Study~2 shows the opposite case in new events, where participation responses on the same side of the market mostly add up to a uniform reduction in price pressure.

\paragraph{Accounting for the change in net flow.}
The exploratory four-term identity (equation~\eqref{eq:fourterm}) divides the change in executed net flow into scale, composition, within-family and interaction terms. At event-level mean counts and executions, the within-family term exceeds the composition term in absolute value in all six events of both samples ($D_e<0$), and the scale term is large. Most of the change is thus attributed to fewer active decisions and to less net buying per active decision, rather than to reallocation across families with different directions. This ordering is unchanged with zero-filled family means and with a Shapley allocation of the interaction term, but the average of pairwise decompositions gives $D_e>0$ in three main-sample and two replication events. The ranking of the composition and within-family terms therefore depends on the level of aggregation, and we do not claim that either dominates. The event-level pattern differs from Study~1's layoff contrast, in which the composition term accounted for 108 of the 109-order change in net submitted orders, although the two accountings use different flow measures.

\paragraph{Interpretation and scope.}
Before any outcome was examined, we named two outcomes that would limit the Study~1 result to particular text bundles: failure to reject H2, and failure of the active-share shifts to recur in the separately worded texts. The first did not occur. The second occurred in three of the six events, so the composition result generalizes to new events and wordings, but the identity of the families that gain is event- and wording-specific. A third named outcome, family contributions that offset in most events, would have confined the result to representation rather than price pressure. It did not occur either: evaluative framing changes both who trades and, through mostly reinforcing responses, the net flow that reaches the provider. Because presentation is the treatment, sensitivity to the text is the object of study and not a confound. What separates a change in composition from a change in trading volume is that families respond in different proportions, and that is the null the preregistered test rejects: a uniform prompt-sensitivity effect that scaled every family's participation by the same factor would leave active composition unchanged (Proposition~\ref{prop:installed}(b)). The test does not show that the same heterogeneity would arise with larger or newer models, in markets with real fundamentals, or with trading prompts written by practitioners, and we do not claim it. These conclusions concern three open-weight language-model families acting as trading agents, six synthetic themes and two wording realizations, in one market design with a specified supply rule. They do not identify why the models respond to evaluations, and they do not establish effects on price accuracy, which requires an independent value benchmark.

\paragraph{Extension with known value.}
Under the responsive provider, midpoint displacement is a flow measure rather than a pricing error. Appendix~\ref{app:knownvalue} specifies an extension with a binary terminal payoff and a stated signal likelihood, so that the information-conditioned value is known independently of model orders. It would test whether the selection margin predicts a reproducible difference in decision regret or value deviation. Because Proposition~\ref{prop:selection} depends on the errors of nonparticipants as well as participants, the extension would elicit forecast errors from active and inactive agents in a parallel, state-matched task.

\section{Conclusion}
Installed model diversity does not ensure stable active representation. In synthetic announcement markets with a fixed mixture of three language-model families, the presentation of the same event changes which families submit orders, by up to 61 percentage points of active share. Installed shares describe active flow only when families participate at equal rates, and framing changes composition only when families respond in different proportions. Study~1 departs sharply from both conditions. A preregistered test in six new events rejects unchanged composition in a main sample and in an independently worded replication. The same reallocation recurs across the two wordings in three of the six events, and in the other three the families that gain differ, so the identity of the gaining families depends on the event and its wording.

Whether a composition shift moves net order flow depends on how participation responds on the two sides of the market. In the layoff event, explanatory text draws a buy-side family into trading while sell-side families trade less, and net orders reverse sign. In the financing event, the original text draws a buy-side and a sell-side family into trading at the same time, and net orders barely change. The composition contribution to net flow equals the share of active decisions reallocated times the gap in conditional direction between the families that gain and lose representation. In the preregistered events, adding a favorable and an unfavorable evaluation to identical facts lowers every family's average participation and lowers executed net flow in all twelve event-sample cells, and family contributions mostly reinforce one another. Our accounting does not show that the change in composition causes the fall in net flow.

For finance research and for supervisors concerned with correlated AI trading, the results imply that model diversity should be measured among the decisions that reach the market, and, in these experiments, that stable net order flow does not imply stable composition. These are hypotheses for empirical work on real markets, not findings about them. In the market-clearing benchmark, selection determines which model errors enter demand, and its covariance with those errors can improve or worsen price accuracy; neither installed diversity nor an active concentration index determines the sign. The findings concern three open-weight model families acting as trading agents in synthetic markets with a specified supply rule. They do not identify why the models respond to evaluations, and they do not establish an effect on price efficiency or accuracy. Testing that effect requires an independent value benchmark, which the known-value extension in Appendix~\ref{app:knownvalue} is designed to provide.
\appendix

\clearpage
\section{Auxiliary financing and rationale evidence}\label{app:legacy}
\subsection{Financing: polarization and participation}
In the original financing batch at $t=5$, Mistral submits 41 buys and zero sells, Qwen zero buys and 65 sells, and Llama nine buys and 16 sells. Under the original audit's co-action convention, Mistral--Qwen pairs agree in zero of 91 eligible cases. The mean seed-level same-minus-cross gap is 0.7327 across 25 computable seeds: 24 positive, one tie, and zero negative. The one-sided drop-ties sign-test p-value is $5.96\times10^{-8}$; the archived seed bootstrap interval is approximately $[0.62,0.84]$.

At $t=12$, the original interpretive follow-up elicits 58 Mistral, 19 Qwen, and 15 Llama active decisions. Every active Mistral and Qwen decision is a buy. The original seed-level co-action gap is 0.1444 across 18 computable seeds, comprising three positives and 15 ties. The drop-ties one-sided p-value is 0.125. This is not strong evidence of a same-family premium at the follow-up node.

The matched follow-up arms further separate participation from direction (Table~\ref{tab:follow}). A factual restatement substantially reduces Mistral activity relative to the interpretive message while increasing Qwen activity. For Mistral, interpretive minus factual participation is positive in all 27 non-tied seeds (two-sided $p=1.49\times10^{-8}$). For Qwen, factual minus no-update participation is positive in all 22 non-tied seeds (two-sided $p=4.77\times10^{-7}$). Mistral and Qwen active decisions remain buy-side in these follow-up arms. Text therefore changes which family supplies flow even when its conditional direction does not change.

\begin{table}[htbp]\centering
\caption{Follow-up participation at $t=12$, original financing study}\label{tab:follow}
\begin{tabular}{lrrrr}\toprule
Text arm & Mistral & Qwen & Llama & Total\\\midrule
Interpretive &58&19&15&92\\
Factual restatement &2&47&12&61\\
No update &0&0&7&7\\\bottomrule
\end{tabular}
\par\smallskip\footnotesize Active cleaned decisions; 150 opportunities per family and arm. These are announcement-node counts, not three-round-window counts.
\end{table}

\subsection{Repeated sampling, mandates, and the content of holds}\label{sec:selection}
The announcement pattern recurs across four batches using the same text and the same 30 structural configurations. Table~\ref{tab:replications} reports all four separately. The first three batches differ in the follow-up assigned at $t=12$, after the announcement measured here; Design A original is a later run. Mistral buys and Qwen sells in every active announcement decision in all four batches. Llama's split varies. This establishes stability under repeated generation for this particular announcement, not robustness to arbitrary wording or independent replication across events.

\begin{table}[htbp]\centering\small\singlespacing
\caption{Repeated announcement batches: buys / sells at $t=5$}\label{tab:replications}
\begin{tabular}{lrrr}\toprule
Batch & Mistral & Qwen & Llama\\\midrule
Original / interpretive follow-up &41 / 0&0 / 65&9 / 16\\
Factual-follow-up batch &36 / 0&0 / 68&5 / 14\\
No-update-follow-up batch &41 / 0&0 / 68&15 / 12\\
Design A original &40 / 0&0 / 66&10 / 13\\\bottomrule
\end{tabular}
\par\smallskip\footnotesize Each cell has 150 decision opportunities; holds are omitted. Configurations are reused across rows, so these are repeated generation batches rather than 120 independent event configurations. Counts are independently reconstructed from the raw logs.
\end{table}

Participation varies with mandates within a family. In the original batch, Mistral agents assigned capital preservation never act at either event round (0/51 at each), whereas opportunistic Mistral agents act in 26/41 opportunities at $t=5$ and 35/41 at $t=12$. Eighteen of Qwen's 19 follow-up actions come from its opportunistic agents. These mandate-conditioned frequencies are descriptive: the experiment does not separately randomize mandate effects within an identical realized portfolio state. At the original no-update node $t=20$, only eight Llama and two Mistral decisions are active, with none for Qwen. This is a within-run negative control, not an independent treatment estimate.

The rationale audit supplies a direct warning against reading conditional direction as a family belief. A stratified sample contains ten holds per family at each event round, for 60 holds in total. A single scorer classifies the accompanying text as neutral, unwilling-bearish, constrained-bullish, or other. Table~\ref{tab:holds} reports the supplied labels, after matching every sampled row to its raw decision log.

\begin{table}[htbp]\centering\small\singlespacing
\caption{Human classification of sampled hold rationales}\label{tab:holds}
\begin{tabular}{lrrr}\toprule
Rationale category & $t=5$ & $t=12$ & Total\\\midrule
Neutral &19&29&48\\
Unwilling-bearish &10&0&10\\
Constrained-bullish &1&0&1\\
Other &0&1&1\\\bottomrule
\end{tabular}
\par\smallskip\footnotesize Stratified sample of ten holds per family and event round; 60 observations, one scorer. These are counts in the selected sample, not population prevalence estimates. Categories describe generated text, not verified internal beliefs or causal motives.
\end{table}

Seven of the ten sampled Mistral announcement holds are labeled unwilling-bearish; two are neutral and one constrained-bullish. Thus, bearish rationale text occurs within a family whose 41 active announcement decisions are all buys. That coexistence is the relevant finding. It does not identify the counterfactual action an agent would take under another mandate, prove that the mandate caused abstention, or reveal internal reasoning. A second independent coder and a prespecified coding rubric would be needed to evaluate classification reliability; neither is available in the supplied audit.

Textual duplication also differs across families. Qwen's 65 announcement sells contain 44 unique reason strings, with one string repeated exactly 17 times; Mistral's 41 announcement reasons are all unique. This is evidence of repeated output phrasing and a warning for text-similarity measures, not direct observation of a common latent computation. Together with the hold sample, it supports treating rationales as auxiliary descriptive evidence rather than as measured beliefs.

\subsection{Rewriting the financing announcement}
Design A reruns the original text alongside de-emotionalized text. Table~\ref{tab:fin} reports this matched batch, rather than borrowing original-batch counts. Total announcement activity is 129 versus 66; the paired difference is 2.10 actions per seed, with an archived bootstrap interval of $[1.47,2.73]$. Mistral and Qwen account for the increase under the original text, while Llama is more active under factual wording.

\begin{table}[htbp]\centering
\caption{Financing announcement: active direction and participation}\label{tab:fin}
\begin{tabular}{llrrrr}\toprule
Text & Family & Buy & Sell & Hold & Active / 150\\\midrule
Original & Mistral &40&0&110&40\\
Original & Qwen &0&66&84&66\\
Original & Llama &10&13&127&23\\\addlinespace
Factual & Mistral &12&0&138&12\\
Factual & Qwen &2&0&148&2\\
Factual & Llama &6&46&98&52\\\bottomrule
\end{tabular}
\end{table}

The two Qwen factual-text trades are buys, compared with 66 sells under original text. This is evidence of strong participation sensitivity; two selected trades cannot establish a stable directional reversal. Mistral buys in both arms. Llama's announcement-round sell share among active decisions rises from $13/23$ to $46/52$. The archive's reason coding identifies dilution or related supply arguments in 37 of those 46 sell explanations, despite removal of the corresponding interpretive vocabulary. Generated reasons illustrate the responses but do not establish the internal computation that caused them.

Under the factual arm, the first-action co-action gap is computable in only eight seeds, versus 25 in the original-text arm. Its mean is 0.688, with six positive differences, one negative, and one tie; the two-sided sign-test p-value is 0.125. A large conditional gap in a small selected subset should not be presented as a precisely estimated population effect.

% ---- moved from main text ----
\subsection{Homogeneous populations and one-sided flow}
The composition experiment replaces the mixed population with 15 Mistral or 15 Qwen agents under the original financing text. Table~\ref{tab:homo} reports submitted order counts over $W_5$. All 159 Mistral orders are buys and all 215 Qwen orders sells. Both homogeneous arms are fully one-sided in all 30 seeds. The mixed control produces both directions in most active seeds and is fully one-sided in four seeds; one mixed seed has no active decisions.

\begin{table}[htbp]\centering
\caption{Composition and order flow in $W_5=\{5,6,7\}$}\label{tab:homo}
\begin{tabular}{lrrrr}\toprule
Population & Buys & Sells & Fully one-sided & Mean $OS$\\\midrule
All Mistral &159&0&30/30&1.000\\
All Qwen &0&215&30/30&1.000\\
Mixed control &52&81&4/30&0.359\\\bottomrule
\end{tabular}
\par\smallskip\footnotesize Mixed $OS$ is defined in 29 seeds. Counts exclude background liquidity and represent submitted decisions, not completed trades or volume.
\end{table}

For either homogeneous arm, the paired mean increase in $OS$ is 0.641 over the 29 defined comparisons. There are 25 positive differences and four ties, giving one-sided $p=2.98\times10^{-8}$. The binary $FULL$ contrast uses all 30 seeds: 26 positives and four ties, giving $p=1.49\times10^{-8}$. These p-values answer different questions.

Qwen increases activity by 2.733 orders per seed relative to mixed control; Mistral increases it by 0.867. The Mistral activity conclusion is sensitive to the tie convention. The common, stronger result is the loss of opposing flow. At the follow-up, however, the mixed market is itself predominantly buy-side, and the $OS$ comparison is not significant under the drop-ties sign test ($p=0.125$). Composition matters in relation to the event-specific differences between families.

This intervention changes both concentration and which models make decisions. It estimates a composition effect in this market. It does not hold average model quality constant, estimate a concentration coefficient independent of family identity, or establish that model diversity always improves outcomes.

\subsection{Configuration uncertainty and the limits of aggregation}
Table~\ref{tab:influence} adds an exploratory resampling and influence analysis of the existing logs. It retains both arms of a configuration together and recomputes pooled active shares in each resample. These are ratios of pooled counts, not averages of market-level shares. The interval target is conditional on each event and the configuration generator; it does not represent uncertainty across future corporate announcements. All six family-share shifts retain their sign under deletion of any single configuration.
\begin{table}[htbp]\centering\small\singlespacing
\caption{Active representation: configuration uncertainty and influence}\label{tab:influence}
\begin{tabular}{llrrr}\toprule
Event & Family & Shift (pp) & 95\% interval & Leave-one-out range\\\midrule
Financing & Mistral & 12.8 & [2.1, 23.0] & [10.2, 15.6]\\
Financing & Qwen & 48.1 & [40.7, 55.3] & [46.8, 50.4]\\
Financing & Llama & -61.0 & [-70.4, -51.2] & [-63.2, -59.0]\\
Layoff & Mistral & 39.9 & [32.8, 47.5] & [38.8, 41.3]\\
Layoff & Qwen & -29.2 & [-39.2, -19.6] & [-30.7, -27.2]\\
Layoff & Llama & -10.7 & [-19.0, -3.0] & [-12.3, -9.3]\\
\bottomrule\end{tabular}
\par\smallskip\footnotesize Original minus factual for financing; explanatory minus factual for layoffs. Shares pool active announcement decisions over 30 paired configurations. Intervals resample complete paired configurations 10,000 times (seed 20260923); they are exploratory pointwise percentile intervals, conditional on the event and configuration generator. Leave-one-out ranges are influence diagnostics, not confidence intervals. No event-population or mediation inference is implied.
\end{table}

Large share shifts need not imply large changes in net flow. Financing net buy-minus-sell counts equal $-29$ under original and $-26$ under factual text, a difference of $-0.10$ per market. Layoff counts equal 55 and $-54$, a difference of 3.63 per market. The corresponding exploratory paired-bootstrap intervals are [$-0.77$,0.63] and [2.80,4.47]. The financing interval is not an equivalence test. These count outcomes exclude quantities and executions. This contrast rules out using the magnitude of active-share movement alone as a sufficient statistic for a market consequence. Section~\ref{sec:netflow} identifies which feature of the reallocation determines its effect on net flow.

Sparse conditional directions require additional restraint. In 1,256 of 10,000 financing paired-configuration bootstrap resamples, at least one family-arm cell has no active decisions. Active family shares remain defined at the pooled-arm level, but a conditional-direction decomposition is undefined in those draws. We therefore do not silently discard those draws to create a decomposition confidence interval. The financing conditional-direction component remains descriptive. No such empty family-arm cell occurs in the layoff resamples; that fact still does not make its accounting decomposition a causal mediation estimate.

As an additional denominator check, the replication script reports the joint family-by-buy/sell/hold distribution using all assigned announcement opportunities. This outcome is always defined, including for inactive markets, but combines participation and conditional direction. It supplements rather than replaces active representation. All newly added checks are exploratory analyses of an inspected archive.

\subsection{Composition and opposing flow}\label{sec:onesided}
Composition also determines whether submitted orders meet orders on the other side of the market. In market $s$, one-sidedness $OS_s=|B_s-S_s|/(B_s+S_s)$ equals the absolute mean direction of its active decisions, $|\sum_f\widetilde w_{fs}d_{fs}|$, so opposing flow disappears when active weight concentrates on families that share an extreme direction. The following benchmark shows that this can happen without any coupling between agents.

\begin{proposition}[One-sidedness without residual coupling]
Suppose that, conditional on a fixed cell and a fixed population assignment, agents' actions are independent, with buy, sell, and hold probabilities $b_i,s_i,h_i$, summing to one. Then the probability of nonempty one-sided submitted flow is
\begin{equation}\label{eq:onesided}
\Prb(FULL=1)=\prod_i(1-s_i)+\prod_i(1-b_i)-2\prod_i h_i.
\end{equation}
If every participating agent can only buy, this probability equals $1-\prod_i h_i$. The analogous statement holds for sell-only agents.
\end{proposition}
\begin{proof}
The no-sell event has probability $\prod_i(1-s_i)$; removing the all-hold event gives buy-only nonempty flow. The symmetric expression gives sell-only nonempty flow. These two events are disjoint, yielding the formula. If all $s_i=0$, then $1-b_i=h_i$ and the formula simplifies as stated.
\end{proof}

This proposition is an explanatory benchmark, not an assumption verified for the sequential simulator. It shows that one-sidedness can approach certainty in a homogeneous population with an extreme conditional direction and sufficient participation, without pair-level stochastic coupling. With deterministic positive participation, a population containing both a buy-only and a sell-only family instead has two-sided flow. No unrestricted monotonicity in model diversity follows: two different families can share a direction, and diversity can alter participation as well as direction.

The framework therefore yields a conditional prediction. Composition changes are most likely to remove opposing flow when the removed family supplies a distinct reading and is active. It does not imply that more provider names always generate more effective behavioral diversity. Nor does it determine whether either direction is correct, since the experiments do not provide an independently calibrated fundamental value.

\subsection{Shared readings versus residual dependence}\label{sec:sharedreadings}
Pair-level agreement in Study~1 could reflect families with extreme marginal directions or residual dependence between agents. The distinction matters for the interpretation above, because only the first is a statement about composition. The identities below show that the observed agreement is largely accounted for by family-level directions.

Fix a conditioning cell $c$ containing the event, text, and other relevant state. Among jointly active agents $i$ and $j$, let $X_i=\mathbf{1}\{A_i=1\}$, and define $p_i=\E[X_i\mid D_i=D_j=1,c]$, with $p_j$ similarly defined. A direct identity gives
\begin{equation}\label{eq:agree}
\Prb(X_i=X_j\mid D_i=D_j=1,c)
=p_ip_j+(1-p_i)(1-p_j)+2\operatorname{Cov}(X_i,X_j\mid D_i=D_j=1,c).
\end{equation}
\begin{proof}
The agreement indicator is $X_iX_j+(1-X_i)(1-X_j)=1-X_i-X_j+2X_iX_j$. Taking conditional expectations and substituting $\E[X_iX_j]=p_ip_j+\operatorname{Cov}(X_i,X_j)$ yields the identity.
\end{proof}

The first two terms are the agreement benchmark under conditional independence. If both marginals equal one, agreement equals one even though the covariance is zero. If one equals one and the other zero, agreement equals zero. Consequently, unanimous same-family actions and complete cross-family disagreement can arise without residual dependence.

In empirical implementations, a family buy rate among all active decisions is often substituted for the pair-specific marginal conditional on joint activity. This substitution additionally requires that partner selection not change the relevant marginal. It is a useful descriptive benchmark here, but not an automatic estimate of residual covariance. Differences in mandates, inventories, order of action, and selection into joint activity can violate the substitution.

The archive's financing decomposition reports agreement of 243/243 for same-model extreme-family pairs, 47/47 for aligned cross-model extreme-family pairs, and 0/182 for opposed extreme-family pairs. The corresponding marginal predictions are one, one, and zero. Non-extreme cells show small residuals in both directions. These calculations are consistent with a shared-reading account. They do not prove conditional independence, rule out economically meaningful residual dependence, or constitute a formal equivalence test: the pairs overlap, and the most striking cells have little or no remaining directional variation.

\paragraph{Two directional benchmarks.}\label{sec:benchmarks}
A family-marginal benchmark asks whether observed pair agreement differs from the agreement implied by family buy rates. A pooled-mix permutation benchmark instead asks whether the directions of active decisions are exchangeable across agents within a seed and round. They preserve different objects and answer different questions.

The archived family-marginal calculation gives predicted versus observed agreement of 0.945 versus 0.952 for same-family pairs and 0.278 versus 0.293 for cross-family pairs at the announcement. At the follow-up, the corresponding values are 0.983 versus 1.000 and 0.838 versus 0.838. These are pooled eligible-pair diagnostics, rather than the equal-weighted mean seed gaps used for inference. Their small discrepancies are consistency evidence, not an equivalence test.

For a reproducible pooled-mix diagnostic, we independently permute buy and sell labels among active agents within each seed and round, preserving activity, family assignments, and the realized pooled buy count. We reconstruct the original overlapping-round pair statistic, average within seed, and then average over computable seeds. Table~\ref{tab:permutation} uses 2,000 draws with analysis random seed 20260917 and the upper-tail correction $(1+\#\{T_b\geq T_{obs}\})/(2001)$.

\begin{table}[htbp]\centering\small\singlespacing
\caption{Exploratory permutation benchmark for family-associated direction}\label{tab:permutation}
\begin{tabular}{lrrrr}\toprule
Event window & Computable seeds & Observed gap & Null s.d. & Upper-tail $p$\\\midrule
Announcement, $W_5$ &25&0.7327&0.1015&0.00050\\
Follow-up, $W_{12}$ &18&0.1444&0.0611&0.04498\\\bottomrule
\end{tabular}
\par\smallskip\footnotesize No agents or events are independently added by the permutations. The null assumes direction exchangeability among active agents conditional on seed and round. It is not a randomization test of family assignment or a test isolating residual within-family dependence. The supplied historical audit reports approximately 0.041 at the follow-up; the new reproducible Monte Carlo run reports 0.04498.
\end{table}

The announcement strongly departs from pooled direction exchangeability, while its observed agreement is close to the family-marginal benchmark. Both statements can hold: family-associated marginal direction explains much of the observed pair pattern. Mandates, states, and selective participation may also violate the permutation null, so rejection does not isolate architecture as the causal channel. The follow-up result is weaker and exploratory; its paired sign test remains $p=0.125$. The two p-values address different null hypotheses and should not be selected between to claim confirmation.

\paragraph{Why the conditioning level matters.}
Pooling heterogeneous cells can itself create apparent excess agreement. By the law of total covariance,
\begin{equation}
\operatorname{Cov}(X_i,X_j)
=\E[\operatorname{Cov}(X_i,X_j\mid C)]
+\operatorname{Cov}(\E[X_i\mid C],\E[X_j\mid C]).
\end{equation}
A common cell variable can therefore generate unconditional dependence even when agents are independent within each cell. Conversely, deterministic same-direction actions within a fixed cell need not exhibit a positive covariance at all.

In the factual layoff data, Qwen's window buy probability is 5/16 in cash-first configurations and 0/56 in capacity-first configurations. Pooling these conditions produces an inappropriate benchmark for pairs drawn disproportionately from one condition. The supplied audit reports that the apparent excess agreement in the pooled extreme cells largely disappears after conditioning on clause order. Same-family extreme-cell agreement is 6/6 and 83/83 in the two factual order cells, matching their unit benchmarks. Under explanatory capacity-first wording, same-family agreement is 113/114 compared with a benchmark of 0.993; opposed cross-family agreement is 4/66 compared with 0.067. These remain descriptive diagnostics, not independent pair-level observations for hypothesis tests.

Clause order is the explicitly examined stratification in this diagnostic; the archive does not establish stability with respect to every mandate, inventory, or market-state variable.

The empirical requirement is to condition at the level where relevant response distributions are stable. A stronger residual-dependence analysis would estimate marginals on independent repetitions, evaluate agreement on held-out repetitions, retain seed-level dependence in resampling, and stratify or model relevant private states. The current archive does not supply that fully separated design.

Selection also changes the composition of pooled pair statistics. Although Mistral constitutes one-third of agents in both arms, its share of announcement actions rises from $21/120=17.5\%$ to $112/195=57.4\%$. A pooled same-model co-action rate consequently changes its composition: in the explanatory window, 170 of 239 eligible same-model pairs are Mistral--Mistral; in the factual window, 93 of 108 are Qwen--Qwen. Comparing the two pooled rates without this accounting risks interpreting selection as a change in an intrinsic coordination parameter.

\clearpage
\section{Study 2 protocol, execution and supporting results}\label{app:study2}

\paragraph{Record.}
This appendix gives the Study~2 protocol, its execution and the supporting results. The texts, the confirmatory test, the pilot-selected full horizon of eight rounds, the campaign size of 1,500 pairs per event and truncated execution were fixed in the preregistration frozen on 28 October 2026. The input-equivalence checks passed on each campaign's actual schedule before it started. The supplementary full-horizon sample of 120 pairs per event was completed on 12 November 2026; as the preregistration orders, it is examined after the main analyses and reported separately. The known-value extension has not been run.

\paragraph{Responsive provider.}
The provider's midpoint is $m(Q_t)=P_{\mathrm{ref}}+Q_t/\lambda$, where $Q_t$ is cumulative net agent purchases from the provider and $\lambda>0$ is fixed across substantive arms. Spread and displayed depth are also fixed after calibration. For event-round orders, let $E_f$ be family $f$'s signed provider-executed quantity. Agent-to-agent fills are excluded. Writing $Q=\sum_fE_f$, the pre-to-post-event midpoint displacement is
\begin{equation}\label{eq:displace}
Y=m(Q_{t_e}^{+})-m(Q_{t_e}^{-})=Q/\lambda.
\end{equation}
A nonzero pre-event inventory state means that the pre-event midpoint differs from $P_{\mathrm{ref}}$. Last traded price is a separate diagnostic. Equation~\eqref{eq:displace} is a property of the supply rule, not an independent behavioral prediction or a measure of mispricing.

\paragraph{Texts and semantic review.}
Each theme has a factual inventory. Both arms reproduce its factual sentences, and the evaluative arm adds explicit favorable or unfavorable evaluations. These evaluations are imposed by the treatment definition and need not follow from the facts. Four reviewers independently answered 90 items on the same review page. The items covered factual consistency with the inventory, unsupported claims, ambiguity and impressions across clause orders. They also included whole-text ratings of direction, strength, emotional tone and clarity. An item was flagged for adjudication when reviewers disagreed or when an answer raised an admission concern, including a unanimous ``uncertain'' answer. Flags prompted review but did not exclude a text automatically. Of the 46 flagged items, 45 reflected disagreement and one a concern. Two answers that contradicted the reviewer's own written explanation were clarified through neutral written questions. Every clarity rating equals 5. All six themes were admitted in both arms and both clause orders without wording changes. The 24 admitted texts and the factual inventories were fixed with their hashes before the pilot, and Appendix~\ref{app:stimuli} reproduces them.

The replication texts are second realizations of the same six themes. One of the authors wrote them from the fixed inventories, and they were archived with their hashes on 22 October 2026 before the pilot. Their review page is the first page without its supplementary step, which was not used for admission, and it asks reviewers to explain any clarity rating of 1 or 2. Three of the four reviewers returned the review. The fourth withdrew for lack of time after receiving it and before any review was compared. A dated amendment recorded the withdrawal before any comparison and before any outcome was examined, and the agreed admission rule, including its disagreement and concern criteria, was applied unchanged to the three reviews. Of 78 items, 39 were flagged. One reviewer's explanations for four entries referred to a different theme; a neutral written question established that the answers were meant for the displayed themes, and the misplaced explanations were not treated as supporting reasoning. A second neutral question, about the direction rating of one evaluative sentence, did not receive an unambiguous answer, so the original rating and its flag were retained. On 14 November 2026 all six themes were admitted in both arms and both clause orders without changes (Appendix~\ref{app:study2repl}).

\paragraph{Primary estimand.}
Let $n_{eczf}$ count active announcement-round decisions of family $f$ in configuration pair $c$ of event $e$ under arm $z$, with $M=5$ eligible agents per family and $F=3$ families. A configuration pair runs one structural configuration under both arms. Let $N_{ecz}=\sum_fn_{eczf}$, $\mu_{ezf}=\E_c[n_{eczf}]$ and $\mu_{ezN}=\sum_f\mu_{ezf}$. Define
\begin{equation}
\theta_{ef}=\frac{\mu_{e1f}}{\mu_{e1N}}-
\frac{\mu_{e0f}}{\mu_{e0N}},\qquad
H_0^{H2}:\theta_{ef}=0\quad\text{for every }e,f.\label{eq:counttarget}
\end{equation}
Because every family has the same number of eligible agents, $\mu_{ezf}$ is proportional to family $f$'s participation rate. The null therefore holds in event $e$ exactly when the response ratios $a_f(1)/a_f(0)$ are equal across families (Proposition~\ref{prop:installed}). Positive expected activity in both arms is required to interpret these ratios. The target is the composition of expected active counts, which matches the pooled shares used in Study~1, not the expectation of a random share. Each event vector is reported together with the average within-event total variation. Opposing events are not canceled before magnitudes are taken.

\paragraph{The bounded test.}
There are $n$ configuration pairs per event. Consecutive pairs in a pre-declared order form $n/2$ disjoint blocks, so $n$ is even. Write $A_{czf}=n_{czf}/M$ and $T_{cz}=N_{cz}/(FM)$, suppressing $e$. For block $j=(c,c')$, define
\begin{align}
D_{jf}=\frac{F}{2(F-1)}[&A_{c1f}T_{c'0}-A_{c0f}T_{c'1}\nonumber\\
&+A_{c'1f}T_{c0}-A_{c'0f}T_{c1}].\label{eq:boundedscore}
\end{align}
The score lies in $[-1,1]$. After substituting total counts, the own-family cross products cancel, and each remaining sum contains at most $F-1$ normalized family counts. Under independent configuration pairs with common arm-specific means within an event,
$\E[D_{jf}]=(\mu_{e1f}\mu_{e0N}-\mu_{e0f}\mu_{e1N})/((F-1)M^2)$,
which is zero under equation~\eqref{eq:counttarget}. Pairs with no active decisions remain in the calculation. No share is imputed and no pair is dropped.

For each of the 18 fixed stakes $\eta\in\{\pm.01,\pm.025,\pm.05,\pm.1,\pm.2,\pm.4,\pm.6,\pm.8,\pm.99\}$, form $E_{ef\eta}=\prod_j(1+\eta D_{ejf})$. The test statistic $E_{H2}$ is the equal-weight average of these values over the six events, three families and 18 stakes. Independence across blocks and zero score means give $\E[E_{H2}]=1$ under the global null. Markov's inequality then gives $\Pr(E_{H2}\geq20)\leq.05$, so $p=\min(1,1/E_{H2})$ is a valid p-value. The construction applies bounded-mean betting methods \citep{betting}. Validity requires independent and identically generated configuration pairs within each event, a fixed eligible count, and pairing and stakes fixed before outcomes are seen. The averaging step needs no independence across families or events.

The averaging and the Bonferroni adjustment in the power study are different operations. Averaging is part of the test. Each $E_{ef\eta}$ has expectation one under the global null, and so does any fixed average of them. No multiplicity correction across events, families or stakes is therefore needed. The Bonferroni adjustment belongs to the design analysis only. It makes the Monte Carlo lower bounds for the 63 power estimates hold simultaneously and does not change the test or its level.

\paragraph{Supporting and exploratory analyses.}
Family-specific changes in participation are reported with qualified uncertainty. H1 states that families change participation by equal percentage points. It is exploratory and is reported as estimates with pair-level uncertainty, without a hypothesis test or significance claims. These intervals are pointwise and approximate. During development an $F$-type procedure for H1 rejected up to 7.63\% of null datasets at a nominal 5\% level, in a scenario with clustering confined to a rare family. No significance test of opposite-direction responses is part of the core plan.

\paragraph{Development record.}
Three procedures failed prespecified checks before the bounded test became primary. In synthetic calibration near the share boundary, a nominal 5\% studentized wild bootstrap rejected 53.7\%, 50.7\% and 32.0\% of null datasets at 12, 30 and 60 blocks per event. A score test with an $F$ reference was then the locked primary rule of an independent confirmation at 150, 180 and 210 pairs per event. At each sample size its size bound failed in at least one null scenario, so it selected no sample size.

A conditional parametric bootstrap was specified next. Aggregate calibration does not make each pair's conditional mean zero, so each bootstrap replicate resamples the paired totals. This conditional argument restores the reference distribution to first order but does not establish finite-sample validity in rare-family regimes. Its frozen development run had 61 cells, 37 of them binding, with 2,000 datasets per cell and 599 bootstrap draws. The scenarios included rare-family, asymmetric-dispersion and dependent-composition designs with residual correlations of about .30 and .41. A binding cell passed when the one-sided 95\% Clopper--Pearson upper bound for its rejection rate was at most 7.5\%, which allows at most 130 rejections in 2,000. Three binding cells failed, all with clustering confined to a rare family at 720 pairs. Their rejection counts were 159, 185 and 166 of 2,000 (7.95\%, 9.25\% and 8.30\%). The prespecified rule therefore closed the construction. In the same run the bounded test rejected none of the 2,000 datasets in any of the nine primary alternatives at 180 pairs.

\paragraph{Power study.}
The power protocol fixed the procedure, the nine alternatives of the development run and a grid of seven sample sizes from 180 to 2,000 pairs per event. Each cell used 2,000 independent simulated datasets. Table~\ref{tab:alternatives} gives the population composition of each arm in the frozen generator. The balanced baseline gives each family an equal share of expected active counts. The unequal baseline has shares of .60, .30 and .10, and the rare-family baseline has shares of .60, .37 and .03. Each baseline has one shifted composition. The shift adds .10 to the family with the smallest baseline share and takes .10 from the other two families in proportion to their baseline shares. With equal baseline shares the first family receives it. In every event one arm has the baseline composition. The other arm has the shifted composition in an affected event and the baseline composition otherwise. An opposite direction swaps the compositions of the two arms. No shift is subtracted from the receiving family, so the rare family's share is either .03 or .13.

The dense-same pattern gives the shifted composition to arm 1 in all six events. The dense-opposing pattern gives it to arm 1 in events 1, 3 and 5 and to arm 0 in events 2, 4 and 6. The two-event pattern gives it to arm 1 in event 1 and to arm 0 in event 2, and both arms keep the baseline composition in events 3 to 6. The total-variation distance between the arms is .10 in every affected event. Because the stake set is symmetric around zero, the test statistic is unchanged when the two arms are exchanged. Generation noise is at the medium level, with concentration 10 in both arms.

\begin{table}[htbp]\centering\small\singlespacing
\caption{Arm compositions in the nine power alternatives}\label{tab:alternatives}
\begin{tabular}{lcc}\toprule
Baseline & Baseline composition & Shifted composition\\\midrule
Balanced & (.333, .333, .333) & (.433, .283, .283)\\
Unequal & (.600, .300, .100) & (.533, .267, .200)\\
Rare family & (.600, .370, .030) & (.538, .332, .130)\\\bottomrule
\end{tabular}

\medskip
\begin{tabular}{lcccc}\toprule
 & \multicolumn{4}{c}{Arm with the shifted composition}\\\cmidrule(lr){2-5}
Pattern of change & Event 1 & Event 2 & Events 3 and 5 & Events 4 and 6\\\midrule
Dense, same & Arm 1 & Arm 1 & Arm 1 & Arm 1\\
Dense, opposing & Arm 1 & Arm 0 & Arm 1 & Arm 0\\
Two events & Arm 1 & Arm 0 & Neither & Neither\\\bottomrule
\end{tabular}
\par\smallskip\footnotesize Population shares of the three families, from the frozen development generator used unchanged in the power study. In each event the arm not named has the baseline composition. Family order and arm labels follow the generator. Shifted shares are rounded to three decimals. The exact shifted rare-family values are $.60-.06/.97$, $.37-.037/.97$ and $.13$.
\end{table}

In each cell a one-sided Clopper--Pearson lower bound at level $1-.05/63$ gives simultaneous coverage across the 63 cells. With 2,000 datasets the bound reaches the .80 target at 1,657 or more rejections. The prespecified rule takes the smallest tested sample size whose lower bound reaches .80 in every alternative. Under these nine alternatives that size is 1,500 pairs per event, and Table~\ref{tab:power} reports the rejection counts. The point estimate and the lower bound both reach .80 at 1,000 pairs in eight alternatives and at 1,500 pairs in the balanced two-event alternative.

\begin{table}[htbp]\centering\small\singlespacing
\caption{Rejections of the bounded H2 test in 2,000 simulated datasets}\label{tab:power}
\begin{tabular}{llrrrr}\toprule
 & & \multicolumn{4}{c}{Configuration pairs per event}\\\cmidrule(lr){3-6}
Baseline composition & Pattern of change & 540 & 720 & 1,000 & 1,500\\\midrule
Balanced & Dense, same & 1 & 552 & 2,000 & 2,000\\
Balanced & Dense, opposing & 0 & 534 & 1,998 & 2,000\\
Balanced & Two events & 0 & 52 & 1,400 & 2,000\\\addlinespace
Unequal & Dense, same & 0 & 631 & 2,000 & 2,000\\
Unequal & Dense, opposing & 0 & 621 & 2,000 & 2,000\\
Unequal & Two events & 0 & 23 & 1,713 & 2,000\\\addlinespace
Rare family & Dense, same & 0 & 621 & 2,000 & 2,000\\
Rare family & Dense, opposing & 0 & 505 & 2,000 & 2,000\\
Rare family & Two events & 0 & 8 & 1,875 & 2,000\\\bottomrule
\end{tabular}
\par\smallskip\footnotesize One prespecified run. Every alternative has a total-variation shift of .10 and medium noise. No dataset rejects at 180 or 360 pairs, and all cells at 2,000 pairs reject in every dataset. Simultaneous one-sided Clopper--Pearson lower bounds use level $1-.05/63$; the .80 power target requires at least 1,657 rejections. At 1,000 pairs the balanced two-event cell has a lower bound of .667; at 1,500 pairs every lower bound is .996.
\end{table}

\paragraph{Pilot and horizon.}
The pilot chose between campaign horizons of eight and twelve rounds before any campaign run. It comprised 72 runs of twelve rounds on configurations excluded from the campaign. Each of the six themes contributed three configuration pairs, with both arms and two generation draws per pair. A rule fixed in advance compared three late-round measures with thresholds of 0.40, 0.25 and 0.10. Measure (A) is the excess of active decisions per run and round in rounds 8 to 11 over the pre-event rounds 1 to 4. Measure (B) is the absolute net provider-executed quantity in rounds 8 to 11 less the same quantity in rounds 1 to 4, relative to the gross quantity submitted at the announcement round. Measure (C) is the share of announcement-round order quantity that executes in round 8 or later. Measure (A) is evaluated for each theme and arm, while (B) and (C) are pooled across all 72 runs. Pooling uses absolute within-run net quantities so that opposite movements across runs cannot cancel. The pooled measures cannot detect drift confined to a single theme. Arms are evaluated separately and never compared.

No threshold was exceeded. Measure (A) was at most 0.40 in every theme and arm, pooled (B) was $-0.014$ and pooled (C) was 0.015. Eight rounds is therefore the full horizon of each market, and the campaigns execute its first six rounds (rounds 0 to 5), as adopted below. At the end of round 11, 35.7\% of announcement-round order quantity was still resting in the book, so the decision does not establish complete execution or settlement. The confirmatory outcome remains the announcement round, and rounds 5 to 7 form a descriptive window.

\paragraph{Campaign scope and the truncation option.}
Earlier Study~2 development runs used Ollama 0.33.2. The pilot used 0.40.1 with the same model files and sampling settings. A throughput test on the same machine then replayed the round 0 to 7 requests of 10 pilot runs under 0.40.1. The campaign will use 0.40.1 and the pilot's model files, and the campaign code checks the server version and model files before the first call of each job and refuses any other. In the throughput test, two to four concurrent request streams raised throughput by factors of 1.27 to 1.36. Responses were identical to the pilot's in all 120 calls of a single stream and in 239 of 240, 359 of 360 and 478 of 480 calls with two, three and four streams. Agreement in these calls does not show that the output distributions are identical. In the pilot the first eight rounds of a run took about 229 seconds of model time on the pilot machine. At 1,500 pairs per event the main campaign therefore needs about 1,144 model-hours of that machine's time, or about 1,373 with a 20\% reserve. A wording replication of the same size needs as much again. Prompts contain neither the round number nor the horizon, so prompts and histories through the announcement round do not depend on the horizon. The truncation option generates the complete eight-round schedule of seed assignments, including the resolution of rare seed collisions, and executes only rounds 0 to 5. The executed rounds are then generated by the same process as the first six rounds of a full-horizon run. Rounds 0 to 5 take 75.8\% of the model time of an eight-round run, so the truncated main campaign needs about 867 model-hours. The authors adopted this truncated execution for both campaigns. It is recorded as a protocol amendment, and input equivalence is checked again on the final campaign code before any campaign run. Later-round dynamics come from a supplementary sample of 120 pairs per event at the full horizon (91.5 model-hours, or 109.8 with the reserve), reported descriptively.

\paragraph{Execution platform.}
The campaigns run on NVIDIA RTX A6000 GPUs of a university computing cluster, one server per GPU, with the same Ollama release (0.40.1, in its Linux build) and the same model files. A replay of the round 0 to 7 requests of 12 pilot runs, in three passes of 1,440 calls, identified settings that reproduce the pilot's effective configuration: flash attention enabled, six CPU threads per server, one request at a time per model and a cold start of the three models before each run. Under these settings every response parsed, repeated passes gave identical responses in every call, also after the order of runs was changed, and 45.6\% of responses were identical to the pilot's. Agreement across platforms is not expected to be complete and is not a test of identical output distributions. A replayed truncated run took about 68 to 74 seconds of model time, about 2.3 to 2.5 times faster than on the pilot machine. Both arms of every test block run on the same frozen execution environment. If the declared platform changes during an event, that event is rerun in full on the new platform, and its earlier records are retained and reported but do not enter H2. Completion dates depend on measured throughput and queue access.

\paragraph{Exact accounting of executed flow.}
Let $a_f=E_f/n_f$ be provider-executed signed quantity per active decision. It combines submitted direction, order size and execution and is not simply a buy probability. With bars denoting arm midpoints, the symmetric identity is
\begin{align}
\Delta Q={}&\underbrace{\Delta N\sum_f\bar w_f\bar a_f}_{C_{\mathrm{scale}}}
+\underbrace{\bar N\sum_f\Delta w_f\bar a_f}_{C_{\mathrm{comp}}}\nonumber\\
&+\underbrace{\bar N\sum_f\bar w_f\Delta a_f}_{C_{\mathrm{within}}}
+\underbrace{\tfrac14\Delta N\sum_f\Delta w_f\Delta a_f}_{C_{\mathrm{int}}}.\label{eq:fourterm}
\end{align}
It follows by expanding each arm around the midpoint. Odd-order terms survive in the difference and produce the three first-order terms and the triple interaction. The identity is applied to equally weighted event-level mean counts and executions. The average of pairwise decompositions is reported separately. These nonlinear aggregations need not give the same component values even though each closes exactly. The family contributions and reinforcement ratio of Proposition~\ref{prop:netflow} are computed from executed flow and reported alongside gross and net submitted flow.

An absent family's counterfactual $a_f$ is unobserved. The continuation convention sets it equal to the other arm's value, and absence in both arms sets both to zero. If an entire arm is inactive, the other arm's composition is copied for accounting only and the change is attributed to scale. These conventions are not used to impute shares for H2. Zero-filled family means are reported as a sensitivity, together with a Shapley allocation that assigns $C_{\mathrm{int}}/3$ to each midpoint component. A component-dominance claim that changes under these allocations is convention-dependent, and none of the components is a causal mediation estimate.

\paragraph{Exploratory mechanism comparisons.}
H3 compares $D_e=|C_{\mathrm{comp},e}|-|C_{\mathrm{within},e}|$, with scale and interaction displayed alongside. Six nonzero signs yield a minimum one-sided binomial reference p-value of $1/64$. With a true positive-sign probability of 0.8, an all-six-positive requirement succeeds only about 26\% of the time. Sign summaries are therefore exploratory. H4 compares homogeneous and mixed markets with total model instances balanced across groups. It estimates squared presentation effects by cross-products of independent draw splits and reports signed contrasts and same-arm placebo pairs alongside. A mixed-versus-homogeneous gap does not by itself identify the composition component, since scale and within-family execution also respond. H4 is deferred to later work.

\paragraph{Execution audit and scope.}
Unique order identifiers are stored, because an agent's submission of some order in a round does not show that an older order did not trigger a given fill. In every run of both samples, the provider executions of announcement-round orders equal the change in the provider's state over that round. The results support evaluative-framing effects on active representation and their flow consequences under the specified responsive supply rule. Claims about price informativeness, welfare or real-market provider adoption require further evidence, such as the extension in Appendix~\ref{app:knownvalue}.

\paragraph{Execution and deviations.}
The main campaign ran from 28 October to 11 November 2026 and the replication from 14 to 16 November, on 11 of the declared nodes, all with the same GPU type, driver (575.57.08) and CUDA version (12.9). In the main campaign 47 runs had one counted failure: 37 failed settings checks, 6 interruptions by a shared-storage failure and 4 interruptions when a duplicate job submission was cancelled. In the replication 19 runs had one counted failure: 18 failed settings checks and one technical failure. No run had a second failure, no run was excluded, and every interrupted run was re-executed in full under the prespecified failure rule. One declared node was left out of all submissions after repeated job-setup failures that occurred before any model call, which affected scheduling only. The five prespecified technical checks passed in both campaigns: every scheduled run committed with verified checksums and ledger chains, the committed-seed check of all 1,620,000 calls, one platform declaration and one environment fingerprint, no outstanding stop, and a matching equivalence record. The checks and the flow analyses were run with separate read-only scripts that leave the frozen campaign code unchanged; the flow script was written after the H2 tests and repeats the loader's checks before reading any outcome. The first H2 job on the main sample reached its time limit while loading the data and before writing any result. It was resubmitted unchanged with a longer limit, and each H2 test was computed once. The dated deviation logs of both campaigns are part of the replication archive.

\paragraph{Unblinding and disclosure.}
Main-sample outcomes were opened on 14 November 2026, after the replication texts had been admitted, the main sample was complete and its technical checks had passed. Replication outcomes were opened on 17 November, after its own sample and technical checks were complete. Relative to the frozen protocol, the main texts, admitted on 21 October under four reviews, and the H2 procedure are unchanged. Two linked elements were amended: the replication texts were admitted on three reviews instead of four, and the opening of the main outcomes, which the protocol tied to that admission, followed the amended admission.

\paragraph{Supporting results.}
Table~\ref{tab:s2_composition} gives the active composition of each arm by event and family in both samples. Table~\ref{tab:s2_participation} gives the changes in family participation and total activity with pair-level standard errors, and Table~\ref{tab:s2_h3} the family contributions, reinforcement ratios and the H3 accounting with its sensitivities. All are estimates without tests, and the intervals are pointwise and approximate.

\begin{table}[htbp]\centering\small\singlespacing
\caption{Active composition at the announcement round, by event, family and arm: main sample and wording replication}\label{tab:s2_composition}
\begin{tabular}{llrrrrrr}\toprule
& & \multicolumn{3}{c}{Main} & \multicolumn{3}{c}{Replication}\\\cmidrule(lr){3-5}\cmidrule(lr){6-8}
Event & Family & A & B & $\theta_{ef}$ & A & B & $\theta_{ef}$\\
\midrule
Customer contract & Llama & 0.187 & 0.385 & 0.198 & 0.210 & 0.253 & 0.043\\
 & Qwen & 0.325 & 0.039 & $-$0.286 & 0.243 & 0.000 & $-$0.243\\
 & Mistral & 0.488 & 0.576 & 0.088 & 0.547 & 0.747 & 0.200\\
& \emph{TV} & & & 0.286 & & & 0.243\\\addlinespace
Inventory purchase & Llama & 0.231 & 0.293 & 0.062 & 0.256 & 0.242 & $-$0.014\\
 & Qwen & 0.240 & 0.386 & 0.146 & 0.327 & 0.305 & $-$0.023\\
 & Mistral & 0.529 & 0.321 & $-$0.208 & 0.417 & 0.453 & 0.036\\
& \emph{TV} & & & 0.208 & & & 0.036\\\addlinespace
Patent licence & Llama & 0.354 & 0.429 & 0.075 & 0.351 & 0.259 & $-$0.092\\
 & Qwen & 0.374 & 0.178 & $-$0.196 & 0.383 & 0.468 & 0.085\\
 & Mistral & 0.272 & 0.393 & 0.121 & 0.266 & 0.273 & 0.007\\
& \emph{TV} & & & 0.196 & & & 0.092\\\addlinespace
Operating permit & Llama & 0.192 & 0.478 & 0.286 & 0.520 & 0.494 & $-$0.026\\
 & Qwen & 0.162 & 0.004 & $-$0.158 & 0.000 & 0.008 & 0.008\\
 & Mistral & 0.645 & 0.517 & $-$0.128 & 0.480 & 0.498 & 0.018\\
& \emph{TV} & & & 0.286 & & & 0.026\\\addlinespace
Division sale & Llama & 0.167 & 0.176 & 0.009 & 0.307 & 0.268 & $-$0.039\\
 & Qwen & 0.529 & 0.708 & 0.179 & 0.372 & 0.579 & 0.207\\
 & Mistral & 0.303 & 0.116 & $-$0.188 & 0.321 & 0.152 & $-$0.168\\
& \emph{TV} & & & 0.188 & & & 0.207\\\addlinespace
Product redesign & Llama & 0.121 & 0.088 & $-$0.034 & 0.113 & 0.081 & $-$0.032\\
 & Qwen & 0.104 & 0.035 & $-$0.068 & 0.327 & 0.114 & $-$0.213\\
 & Mistral & 0.775 & 0.877 & 0.102 & 0.560 & 0.805 & 0.244\\
& \emph{TV} & & & 0.102 & & & 0.244\\\addlinespace
\multicolumn{2}{l}{Average within-event TV} & & & 0.211 & & & 0.141\\
\bottomrule
\end{tabular}
\par\smallskip\footnotesize Shares of each arm's pooled mean active count (A factual, B evaluative); $\theta_{ef}$ = B share $-$ A share; TV $=\tfrac12\sum_f|\theta_{ef}|$, the share of active decisions reallocated across families. TV is taken within each event before averaging, so opposite shifts in different events do not cancel. Supporting analysis; no test.
\end{table}

\begin{table}[htbp]\centering\scriptsize\singlespacing\setlength{\tabcolsep}{3pt}
\caption{Change in family participation rates and in total activity (B $-$ A), with pair-level standard errors}\label{tab:s2_participation}
\resizebox{\linewidth}{!}{%
\begin{tabular}{lrrrrrrrr}\toprule
& \multicolumn{2}{c}{Llama} & \multicolumn{2}{c}{Qwen} & \multicolumn{2}{c}{Mistral} & \multicolumn{2}{c}{Total active}\\\cmidrule(lr){2-3}\cmidrule(lr){4-5}\cmidrule(lr){6-7}\cmidrule(lr){8-9}
Event & Main & Repl. & Main & Repl. & Main & Repl. & Main & Repl.\\
\midrule
Customer contract & $-$0.151\,(.0060) & $-$0.196\,(.0066) & $-$0.434\,(.0065) & $-$0.361\,(.0057) & $-$0.510\,(.0058) & $-$0.467\,(.0082) & $-$5.47\,(.054) & $-$5.12\,(.063)\\
Inventory purchase & $-$0.242\,(.0069) & $-$0.168\,(.0074) & $-$0.207\,(.0090) & $-$0.217\,(.0071) & $-$0.741\,(.0067) & $-$0.238\,(.0129) & $-$5.95\,(.071) & $-$3.12\,(.082)\\
Patent licence & $-$0.033\,(.0071) & $-$0.070\,(.0068) & $-$0.192\,(.0054) & 0.077\,(.0108) & 0.011\,(.0091) & 0.010\,(.0069) & $-$1.06\,(.064) & 0.09\,(.083)\\
Operating permit & $-$0.062\,(.0058) & $-$0.136\,(.0057) & $-$0.161\,(.0045) & 0.001\,(.0004) & $-$0.504\,(.0060) & $-$0.120\,(.0048) & $-$3.64\,(.046) & $-$1.27\,(.036)\\
Division sale & $-$0.016\,(.0059) & 0.006\,(.0061) & 0.074\,(.0052) & 0.178\,(.0075) & $-$0.196\,(.0054) & $-$0.079\,(.0043) & $-$0.69\,(.048) & 0.52\,(.056)\\
Product redesign & $-$0.045\,(.0046) & $-$0.082\,(.0051) & $-$0.069\,(.0033) & $-$0.343\,(.0055) & $-$0.048\,(.0063) & $-$0.047\,(.0050) & $-$0.81\,(.044) & $-$2.36\,(.045)\\
Pooled & $-$0.091\,(.0025) & $-$0.108\,(.0026) & $-$0.165\,(.0024) & $-$0.111\,(.0028) & $-$0.331\,(.0027) & $-$0.157\,(.0031) & $-$2.94\,(.023) & $-$1.88\,(.026)\\
\bottomrule
\end{tabular}}
\par\smallskip\footnotesize Participation = active decisions of the family / 5 eligible agents at $t=5$; total active = active decisions per market (of 15). Paired differences over 1,500 configuration pairs per event; standard errors in parentheses (pair level; pooled: equal-weight mean over events, events treated as independent). Approximate 95\% intervals ($\pm1.96$ SE) and arm levels are in the results package. Estimates only; no adjustment for families or events; no test (H1 exploratory).
\end{table}

\begin{table}[htbp]\centering\scriptsize\singlespacing\setlength{\tabcolsep}{3pt}
\caption{Family contributions to the change in executed net flow, reinforcement ratio, and the H3 accounting}\label{tab:s2_h3}
\resizebox{\linewidth}{!}{%
\begin{tabular}{lrrrrrrrrrrrrrr}\toprule
& \multicolumn{3}{c}{$\Delta Q_f$} & & & & \multicolumn{4}{c}{Components} & \multicolumn{3}{c}{$D_e$ sensitivities}\\\cmidrule(lr){2-4}\cmidrule(lr){8-11}\cmidrule(lr){12-14}
Event & Llama & Qwen & Mistral & $\Delta Q$ & $\sum|\Delta Q_f|$ & $R$ & $C_{scale}$ & $C_{comp}$ & $C_{within}$ & $C_{int}$ & $D_e$ & 0-fill & Shapley\\
\midrule
\multicolumn{14}{l}{\emph{Main}}\\
Customer contract & 0.25 & $-$2.98 & $-$4.78 & $-$7.52 & 8.01 & 0.94 & $-$3.22 & $-$1.98 & $-$2.61 & 0.29 & $-$0.63 & $-$0.63 & $-$0.63 & 0.16\\
Inventory purchase & $-$0.83 & $-$3.28 & $-$7.01 & $-$11.12 & 11.12 & 1.00 & $-$4.21 & $-$1.61 & $-$5.64 & 0.34 & $-$4.04 & $-$4.04 & $-$4.04 & $-$0.72\\
Patent licence & 0.32 & $-$2.04 & $-$3.68 & $-$5.40 & 6.04 & 0.89 & $-$0.08 & $-$0.30 & $-$5.04 & 0.02 & $-$4.74 & $-$4.74 & $-$4.74 & 0.72\\
Operating permit & $-$0.67 & $-$1.30 & $-$5.23 & $-$7.21 & 7.21 & 1.00 & $-$4.36 & $-$1.09 & $-$1.65 & $-$0.10 & $-$0.56 & $-$0.56 & $-$0.56 & 0.74\\
Division sale & $-$0.01 & $-$1.79 & $-$2.81 & $-$4.61 & 4.61 & 1.00 & 0.53 & $-$0.88 & $-$4.18 & $-$0.07 & $-$3.30 & $-$3.30 & $-$3.30 & $-$1.45\\
Product redesign & $-$0.10 & $-$0.43 & $-$1.21 & $-$1.74 & 1.74 & 1.00 & $-$1.13 & 0.25 & $-$0.87 & 0.01 & $-$0.62 & $-$0.62 & $-$0.61 & $-$0.66\\
\addlinespace
\multicolumn{14}{l}{\emph{Replication}}\\
Customer contract & 0.48 & $-$2.21 & $-$7.25 & $-$8.97 & 9.94 & 0.90 & $-$6.94 & 1.35 & $-$3.19 & $-$0.20 & $-$1.84 & $-$1.84 & $-$1.97 & $-$1.36\\
Inventory purchase & 0.38 & $-$1.34 & $-$2.49 & $-$3.45 & 4.22 & 0.82 & $-$1.31 & 0.01 & $-$2.17 & 0.02 & $-$2.16 & $-$2.16 & $-$2.14 & 1.49\\
Patent licence & 0.58 & $-$2.07 & 0.17 & $-$1.32 & 2.81 & 0.47 & $-$0.10 & 0.42 & $-$1.65 & $-$0.00 & $-$1.22 & $-$1.22 & $-$1.23 & $-$0.70\\
Operating permit & 0.16 & $-$0.01 & $-$1.21 & $-$1.05 & 1.38 & 0.76 & $-$0.66 & 0.07 & $-$0.46 & $-$0.00 & $-$0.39 & $-$0.39 & $-$0.39 & $-$0.15\\
Division sale & $-$0.31 & $-$1.33 & $-$1.46 & $-$3.09 & 3.09 & 1.00 & $-$0.25 & $-$1.19 & $-$1.68 & 0.02 & $-$0.49 & $-$0.49 & $-$0.49 & 0.37\\
Product redesign & $-$0.09 & $-$2.02 & $-$1.94 & $-$4.05 & 4.05 & 1.00 & $-$3.25 & 0.39 & $-$1.29 & 0.10 & $-$0.90 & $-$0.90 & $-$0.83 & $-$0.77\\
\addlinespace
\bottomrule
\end{tabular}}
\par\smallskip\footnotesize Per market at $t=5$, B $-$ A, executed flow with the provider. $R=|\Delta Q|/\sum_f|\Delta Q_f|$ equals 1 when all family contributions share a sign (reinforcing) and is small when they offset. Four-term identity on equally weighted event-level mean counts and executions, with the continuation conventions; $D_e=|C_{comp}|-|C_{within}|$. Sensitivities: zero-filled family means; Shapley allocation of $C_{int}/3$ to each first-order term. The average of pairwise decompositions and $R$ on net submitted orders are in the results package. Exploratory accounting; no component is a causal mediation estimate.
\end{table}

\clearpage
\section{Specification of the known-value extension}\label{app:knownvalue}
This extension is planned as later work and is not part of Study~2. Under the responsive provider, $Y=Q/\lambda$ is a midpoint identity. The known-value extension uses a binary terminal payoff with a stated prior and signal likelihood, so that the information-conditioned expected payoff $v^*$ is known independently of model orders. Equivalent presentations retain the same payoff and likelihood information. Agents face disclosed cash, inventory, fees, a five-share cap and a stated mean--variance objective. The evaluator enumerates feasible integer orders at the prevailing execution schedule and records utility regret relative to the best feasible decision.

For a provider with slope $1/\lambda$, a signed order $q$ at pre-order midpoint $m$ costs
\begin{equation}
C(q;m)=mq+\frac{q^2}{2\lambda}+k|q|,
\qquad m^+=m+q/\lambda.\label{eq:cost}
\end{equation}
A negative $q$ generates sale proceeds net of impact and fees. With cash $c$, inventory $h$, payoff variance $\sigma_V^2$ and disclosed risk coefficient $\gamma$, the order criterion is
\begin{equation}
U(q)=c-C(q;m)+(h+q)v^*-
\frac{\gamma}{2}(h+q)^2\sigma_V^2.\label{eq:utility}
\end{equation}
The finite feasible set $\mathcal Q$ includes zero and enforces cash, inventory, positive marginal execution prices and the cap. Regret is $\max_{q\in\mathcal Q}U(q)-U(q_{\mathrm{deployed}})$. Invalid requests deploy a hold and are flagged separately. Regret measures a conditional decision loss at the observed state, not aggregate welfare or the globally optimal dynamic trading policy.

The reference implementation supplies posterior values, exact enumeration, execution accounting and non-LLM test policies. It has not generated LLM evidence and is separate from the Study~2 provider implementation. Because the design uses common information and risk preferences, the endpoint $(m-v^*)^2$ is an explicitly defined laboratory value-deviation measure; it is not automatically an equilibrium pricing error in real markets. It is compared with a common-information oracle-policy path under identical initial conditions, and signed deviations, regret and feasibility are reported separately. A risk-sensitive oracle can optimally reduce inventory and move the midpoint away from the posterior mean, so value deviation alone is not a sufficient measure of decision quality. A separate risk-neutral panel is reported, and the stated objective is fixed within every comparison. Transaction gains can be transfers from the provider and are not social surplus.

The critical test is whether the presentation-induced selection margin predicts a reproducible difference in decision loss or value deviation in this environment. Forecast errors are elicited for active and inactive agents in a parallel, state-matched measurement task, which avoids feedback into the main trading path. Mechanical replay of orders can diagnose accounting channels but cannot identify a natural indirect effect. A later randomized eligibility intervention would have its own total policy effect, because participation rules also change incentives.

\clearpage
\section{Data provenance and reproducibility}
All reported results are computed from archived simulation logs. The archive is organized as follows.

\begin{longtable}{>{\raggedright\arraybackslash}p{0.25\linewidth}>{\raggedright\arraybackslash}p{0.66\linewidth}}
\toprule
Evidence & Archive location and treatment\\\midrule\endhead
Original financing & \texttt{data/paired\_hetero\_split\_seed*.csv}; original co-action results documented in \texttt{docs/VERIFIED\_RESULT\_TABLE.md}.\\
Follow-up controls & \texttt{data/paired\_hetero\_neutral\_seed*.csv} and \texttt{data/paired\_hetero\_noinfo\_seed*.csv}.\\
Design A & \texttt{data/design\_a\_30/}; reproduced with \texttt{analysis.analyze\_design\_a\_30seeds}.\\
Homogeneous arms & \texttt{data/homo\_arms/}; reproduced with \texttt{analysis.analyze\_homogeneous\_arms}; mixed control is Design A original.\\
Layoff study & \texttt{data/theme\_layoff/}; reproduced with \texttt{analysis.analyze\_theme\_layoff -{}-with-financing}.\\
Marginal decomposition & \texttt{docs/SHARED\_MODEL\_VS\_SHARED\_READING.md} and \texttt{docs/THEME\_LAYOFF\_AUDIT.md}; reported as descriptive diagnostics.\\
Net-flow accounting & Table~\ref{tab:flowaccounting}; computed from the announcement-round counts in Tables~\ref{tab:fin} and~\ref{tab:layoff}.\\
Verification & \texttt{manuscript/verify\_revision.py} and its CSV/JSON outputs.\\
Configuration robustness & \texttt{manuscript/review\_robustness.py}; paired-configuration robustness results.\\
Study~2 texts & \nolinkurl{study2/texts/STUDY2_ADMITTED_TEXTS_2026-10-21_release.json} (24 admitted texts, factual inventories and evaluative sentences with the hash of each string, reproduced in Appendix~\ref{app:stimuli}; copied from the frozen admission record, whose hash it gives)\newline \nolinkurl{study2/texts/REPLICATION_TEXTS_FROZEN_2026-10-22.json} (replication texts with hashes, archived on 22 October 2026 and admitted on 14 November 2026, reproduced in Appendix~\ref{app:study2repl})\\
Study~2 semantic review & \nolinkurl{study2/review/RECONCILIATION_RULE_four_reviewers_draft2_agreed.md} (agreed admission rule)\newline Reviewer exports, the compiled item table and the clarification records name the reviewers and are privately archived. Their anonymized release is planned, subject to the reviewers' consent.\\
Study~2 test development & \nolinkurl{validation/INFERENCE_CALIBRATION.md} (wild bootstrap)\newline \nolinkurl{study2/confirmation/H2_LOCKED_FOR_CONFIRMATION.md}\newline \nolinkurl{study2/confirmation/CONFIRMATION_RESULT.md} (score test)\newline \nolinkurl{study2/development/H2_bootstrap_specification_draft.md} (draft 5)\newline \nolinkurl{study2/development/DEVELOPMENT_MANIFEST.json}\newline \nolinkurl{study2/development/DEVELOPMENT_RESULTS.md}\newline \nolinkurl{study2/development/DEVELOPMENT_DECISION.json}\\
Study~2 power study & \nolinkurl{study2/power/BOUNDED_POWER_PROTOCOL.md} (draft 2)\newline \nolinkurl{study2/power/BOUNDED_POWER_MANIFEST.json}\newline \nolinkurl{study2/power/FREEZE_RECORD.json}\newline \nolinkurl{study2/power/BOUNDED_POWER_RESULT.json}\newline \nolinkurl{study2/power/run_bounded_power.py} (generator and test code identified by hash in the freeze record)\newline \nolinkurl{study2/power/COMPOSITION_VECTORS_frozen_generator_2026-10-24.json} (Table~\ref{tab:alternatives})\\
Study~2 preregistration and campaigns & \nolinkurl{study2/preregistration/PREREGISTRATION_draft9r2.md} (frozen 28 October 2026)\newline \nolinkurl{study2/campaign/} (campaign code, manifests, job lists and equivalence records)\newline \nolinkurl{study2/d8/D8_RECORD.json} (admission of the replication texts)\newline \nolinkurl{study2/deviations/} (dated deviation logs and the three-review amendment)\\
Study~2 results & \nolinkurl{study2/results/H2_main.json}, \nolinkurl{study2/results/H2_replication.json} (confirmatory tests)\newline \nolinkurl{study2/results/TECHNICAL_CHECKS_*_SUMMARY.json} (completed-sample checks)\newline \nolinkurl{study2/results/FLOW_main.json}, \nolinkurl{study2/results/FLOW_replication.json} (composition, participation, flow accounting and H3; computed by \nolinkurl{study2/results/code/study2_flow_descriptive.py})\newline \nolinkurl{study2/results/code/make_results_package.py} (Tables~\ref{tab:s2_h2}--\ref{tab:s2_h3} and Figure~\ref{fig:s2composition})\\
Study~2 pilot & \nolinkurl{study2/pilot/PILOT_HORIZON_RULE_draft.md} (draft 7)\newline \nolinkurl{study2/pilot/horizon_decision.py}\newline \nolinkurl{study2/pilot/HORIZON_SCRIPT_FREEZE.json}\newline \nolinkurl{study2/pilot/horizon_decision.json}\newline \nolinkurl{study2/pilot/SERVER_VERSIONS_2026-10-24_r3.json} (server version of each batch and the planned campaign pin)\newline \nolinkurl{study2/pilot/runs/} (per-theme manifests and run records)\\
\bottomrule
\end{longtable}

\noindent To reproduce the reported analyses from the project root:
\begin{verbatim}
python -m analysis.analyze_design_a_30seeds
python -m analysis.analyze_homogeneous_arms
python -m analysis.analyze_theme_layoff --with-financing
python manuscript/verify_revision.py
python manuscript/analyze_composition.py
python manuscript/review_robustness.py
python manuscript/make_figures.py
\end{verbatim}
\begingroup\setstretch{1.15}
An integrity audit (\texttt{manuscript/verify\_merge.py}) covers 308 run files and 115,500 decisions, including pilot runs used only for integrity and realized-price checks; main estimates exclude pilots and use the arm-specific samples described above. \texttt{manuscript/make\_figures.py} regenerates the numerical figures without model calls. Software and synthetic validation checks are in \texttt{validation/}. The Study~2 records above document the design decisions, synthetic computations, horizon pilot, campaigns and results (Appendix~\ref{app:study2}). The Study~2 tables and figure are rebuilt from the result files without model calls.
\endgroup
\clearpage

\section{Exact prompt and treatment texts}\label{app:stimuli}
The Study~1 prompt below is extracted from \texttt{sim/agents.py}. The Study~2 pilot and campaigns use the same prompt followed by one paragraph stating the order constraints. This campaign prompt was frozen with the preregistration on 28 October 2026 and is reproduced in full after the Study~1 prompt. Braces denote the Python format fields used by the runner; doubled braces surround the requested JSON. The instruction that a decision is independent is a prompt instruction, not a statistical independence result.
\begin{singlespace}
\begin{Verbatim}
You are an independent retail trader. You do not know, see, or talk to any other trader. Every decision is your own, based on your mandate and your own reading of the situation.

Your mandate: {mandate}.
Cash: {cash:.2f}. Shares of SmallCo held: {shares}.
Recent price: {last_price:.2f}. Top bid size: {top_bid}, top ask size: {top_ask}.
Latest news: "{news}"

Most rounds carry no real information, and you are never required to trade. Holding is a perfectly reasonable choice. Do not trade out of boredom, and do not react to small price moves by themselves.

If there is news, decide for yourself whether it actually matters, which way it points, and whether it is strong enough to act on given your mandate and your current position. Choose the side and the size yourself. There is no preset correct reaction to any piece of news.

If you place an order, pick a limit price that could realistically execute against the current book: for a buy, at or slightly above the recent price; for a sell, at or slightly below it.

Decide ONE action. Respond ONLY with JSON, no other text:
{{
  "action": "buy" | "sell" | "hold",
  "limit_price": number,
  "qty": integer,
  "reason": "one or two sentences"
}}

If you hold, set qty 0.
\end{Verbatim}
\end{singlespace}
\subsection{Study 2 campaign prompt}
The Study~2 campaign prompt is the Study~1 prompt above followed by the order-constraints paragraph. It is reproduced in full from the frozen \texttt{sim/agents.py}.
\begin{singlespace}
\begin{Verbatim}
You are an independent retail trader. You do not know, see, or talk to any other trader. Every decision is your own, based on your mandate and your own reading of the situation.

Your mandate: {mandate}.
Cash: {cash:.2f}. Shares of SmallCo held: {shares}.
Recent price: {last_price:.2f}. Top bid size: {top_bid}, top ask size: {top_ask}.
Latest news: "{news}"

Most rounds carry no real information, and you are never required to trade. Holding is a perfectly reasonable choice. Do not trade out of boredom, and do not react to small price moves by themselves.

If there is news, decide for yourself whether it actually matters, which way it points, and whether it is strong enough to act on given your mandate and your current position. Choose the side and the size yourself. There is no preset correct reaction to any piece of news.

If you place an order, pick a limit price that could realistically execute against the current book: for a buy, at or slightly above the recent price; for a sell, at or slightly below it.

Decide ONE action. Respond ONLY with JSON, no other text:
{{
  "action": "buy" | "sell" | "hold",
  "limit_price": number,
  "qty": integer,
  "reason": "one or two sentences"
}}

If you hold, set qty 0.

Order constraints: For a buy or sell, request an integer quantity from 1 to 5 shares. Do not borrow cash or sell shares you do not hold. At submission, a buy's quantity times its limit price must not exceed the cash shown above, and a sell's quantity must not exceed the shares shown above. These checks use your currently shown cash and holdings; outstanding unmatched orders remain active and do not reserve cash or shares. If you hold, set qty to 0.
\end{Verbatim}
\end{singlespace}
\subsection{Study 1 recorded news variants}
The following strings are extracted from the archived decision logs at rounds 0, 5, 12, and 20. Identical strings share an identifier. The machine-readable manifest lists arm, round, source file, and a newly recomputed SHA-256 hash; these hashes verify this package's text, not historical registration timing.

\paragraph{News 01}
\noindent\textit{Used in:} Original financing / interpretive follow-up, $t=0$; Factual follow-up, $t=0$; No-update follow-up, $t=0$; Design A original, $t=0$; Design A factual, $t=0$; Layoff factual, cash-first, $t=0$; Layoff factual, capacity-first, $t=0$; Layoff explanatory, cash-first, $t=0$; Layoff explanatory, capacity-first, $t=0$.

\begin{quote}\singlespacing
Trading opens. No major news.
\end{quote}

\paragraph{News 02}
\noindent\textit{Used in:} Original financing / interpretive follow-up, $t=5$; Factual follow-up, $t=5$; No-update follow-up, $t=5$; Design A original, $t=5$.

\begin{quote}\singlespacing
SmallCo has announced a new equity financing round. The deal removes near-term funding pressure and gives the company enough cash to keep pursuing its expansion plan, but the shares were issued at a low valuation and will increase the share count by about 30\%. Early trading is split: some traders are buying after focusing on the reduced funding risk, while others are selling because the dilution permanently lowers existing shareholders claim on each share.
\end{quote}

\paragraph{News 03}
\noindent\textit{Used in:} Original financing / interpretive follow-up, $t=12$; Design A original, $t=12$; Design A factual, $t=12$.

\begin{quote}\singlespacing
The market remains divided over SmallCo financing. Buyers are treating the new cash as protection against a funding shortfall and a chance to keep the growth plan alive. Sellers are focusing on the low valuation and the 30\% increase in share count, arguing that the company may survive but existing shareholders now own a smaller claim. Market participants expect the stock to reprice today as investors decide which effect matters more.
\end{quote}

\paragraph{News 04}
\noindent\textit{Used in:} Original financing / interpretive follow-up, $t=20$; Factual follow-up, $t=20$; No-update follow-up, $t=12$; No-update follow-up, $t=20$; Design A original, $t=20$; Design A factual, $t=20$; Layoff factual, cash-first, $t=20$; Layoff factual, capacity-first, $t=20$; Layoff explanatory, cash-first, $t=20$; Layoff explanatory, capacity-first, $t=20$.

\begin{quote}\singlespacing
No further update has been released by SmallCo.
\end{quote}

\paragraph{News 05}
\noindent\textit{Used in:} Factual follow-up, $t=12$.

\begin{quote}\singlespacing
Trading continues following SmallCo's financing announcement, with both buy and sell orders present in the market. The financing removed the company's near-term funding pressure and increased the share count by about 30\%.
\end{quote}

\paragraph{News 06}
\noindent\textit{Used in:} Design A factual, $t=5$.

\begin{quote}\singlespacing
SmallCo has announced a new equity financing round. The financing provides the company with cash sufficient to fund its operations and its planned expansion. The new shares were issued at a price below the recent market price, and the total number of shares outstanding will increase by about 30\%.
\end{quote}

\paragraph{News 07}
\noindent\textit{Used in:} Layoff factual, cash-first, $t=5$.

\begin{quote}\singlespacing
SmallCo has announced a reduction of its workforce by about 20\%. The company expects the change to lower annual operating costs by roughly 25\% and to extend its cash runway from about 12 months to more than 24 months. The reduction eliminates one of the company's three product development teams, and it follows two consecutive quarters in which sales grew more slowly than the company had planned.
\end{quote}

\paragraph{News 08}
\noindent\textit{Used in:} Layoff factual, cash-first, $t=12$; Layoff factual, capacity-first, $t=12$; Layoff explanatory, cash-first, $t=12$; Layoff explanatory, capacity-first, $t=12$.

\begin{quote}\singlespacing
SmallCo confirmed the workforce reduction announced earlier. No additional details were released.
\end{quote}

\paragraph{News 09}
\noindent\textit{Used in:} Layoff factual, capacity-first, $t=5$.

\begin{quote}\singlespacing
SmallCo has announced a reduction of its workforce by about 20\%. The reduction eliminates one of the company's three product development teams, and it follows two consecutive quarters in which sales grew more slowly than the company had planned. The company expects the change to lower annual operating costs by roughly 25\% and to extend its cash runway from about 12 months to more than 24 months.
\end{quote}

\paragraph{News 10}
\noindent\textit{Used in:} Layoff explanatory, cash-first, $t=5$.

\begin{quote}\singlespacing
SmallCo has announced a reduction of its workforce by about 20\%. The move relieves near-term cost pressure: the company expects annual operating costs to fall by roughly 25\% and its cash runway to extend from about 12 months to more than 24 months. But the reduction eliminates one of the company's three product development teams, and it follows two consecutive quarters in which sales grew more slowly than the company had planned. The lost development capacity permanently reduces what the company can build. Early trading is split: some traders are buying after focusing on the longer runway, while others are selling because of the lost capacity.
\end{quote}

\paragraph{News 11}
\noindent\textit{Used in:} Layoff explanatory, capacity-first, $t=5$.

\begin{quote}\singlespacing
SmallCo has announced a reduction of its workforce by about 20\%. The reduction eliminates one of the company's three product development teams, and it follows two consecutive quarters in which sales grew more slowly than the company had planned. The lost development capacity permanently reduces what the company can build. But the move relieves near-term cost pressure: the company expects annual operating costs to fall by roughly 25\% and its cash runway to extend from about 12 months to more than 24 months. Early trading is split: some traders are selling because of the lost capacity, while others are buying after focusing on the longer runway.
\end{quote}

\subsection{Study 2 admitted texts}\label{app:study2texts}
These 24 texts were admitted on 21 October 2026 and are fixed. For each theme the factual inventory (F1 to F3) and the evaluative sentences (E1 and E2) are listed first. The four texts follow in both clause orders. The file \nolinkurl{study2/texts/STUDY2_ADMITTED_TEXTS_2026-10-21_release.json} contains the same strings with the SHA-256 hash of the UTF-8 bytes of each. The first 12 hexadecimal digits of each text's hash are shown with the text. The hashes refer to the strings in that file, since typesetting changes apostrophes and line breaks.

\paragraph{Theme 1. Customer contract and unit margin}
\begin{description}\singlespacing
\item[F1] SmallCo has signed a supply contract covering the next 18 months.
\item[F2] The contract raises committed unit orders by 30\% relative to the current order book.
\item[F3] The contracted unit price is 12\% below the current average selling price, while the stated unit production cost is unchanged.
\item[E1] The larger order book is favorable.
\item[E2] The lower contracted unit price is unfavorable.
\end{description}
\noindent\textit{Factual arm, clause order A} \hfill {\footnotesize\texttt{e76ea9f390db}}
\begin{quote}\singlespacing
SmallCo has signed a supply contract covering the next 18 months. The contract raises committed unit orders by 30\% relative to the current order book. The contracted unit price is 12\% below the current average selling price, while the stated unit production cost is unchanged.
\end{quote}
\noindent\textit{Evaluative arm, clause order A} \hfill {\footnotesize\texttt{103ede7ef500}}
\begin{quote}\singlespacing
SmallCo has signed a supply contract covering the next 18 months. The contract raises committed unit orders by 30\% relative to the current order book. The larger order book is favorable. The contracted unit price is 12\% below the current average selling price, while the stated unit production cost is unchanged. The lower contracted unit price is unfavorable.
\end{quote}
\noindent\textit{Factual arm, clause order B} \hfill {\footnotesize\texttt{430e7c5681a1}}
\begin{quote}\singlespacing
SmallCo has signed a supply contract covering the next 18 months. The contracted unit price is 12\% below the current average selling price, while the stated unit production cost is unchanged. The contract raises committed unit orders by 30\% relative to the current order book.
\end{quote}
\noindent\textit{Evaluative arm, clause order B} \hfill {\footnotesize\texttt{128561a24406}}
\begin{quote}\singlespacing
SmallCo has signed a supply contract covering the next 18 months. The contracted unit price is 12\% below the current average selling price, while the stated unit production cost is unchanged. The lower contracted unit price is unfavorable. The contract raises committed unit orders by 30\% relative to the current order book. The larger order book is favorable.
\end{quote}

\paragraph{Theme 2. Supplier discount and cash commitment}
\begin{description}\singlespacing
\item[F1] SmallCo has accepted a bulk purchase agreement for its main production input.
\item[F2] The agreement reduces the input price per unit by 15\% and covers 9 months of planned input use.
\item[F3] Payment is due immediately and equals 25\% of current cash; the previous purchase schedule spread the same input quantity over monthly payments.
\item[E1] The lower input price is favorable.
\item[E2] The immediate cash commitment is unfavorable.
\end{description}
\noindent\textit{Factual arm, clause order A} \hfill {\footnotesize\texttt{0caea8ad6752}}
\begin{quote}\singlespacing
SmallCo has accepted a bulk purchase agreement for its main production input. The agreement reduces the input price per unit by 15\% and covers 9 months of planned input use. Payment is due immediately and equals 25\% of current cash; the previous purchase schedule spread the same input quantity over monthly payments.
\end{quote}
\noindent\textit{Evaluative arm, clause order A} \hfill {\footnotesize\texttt{f12da90b03ed}}
\begin{quote}\singlespacing
SmallCo has accepted a bulk purchase agreement for its main production input. The agreement reduces the input price per unit by 15\% and covers 9 months of planned input use. The lower input price is favorable. Payment is due immediately and equals 25\% of current cash; the previous purchase schedule spread the same input quantity over monthly payments. The immediate cash commitment is unfavorable.
\end{quote}
\noindent\textit{Factual arm, clause order B} \hfill {\footnotesize\texttt{06dccd78e923}}
\begin{quote}\singlespacing
SmallCo has accepted a bulk purchase agreement for its main production input. Payment is due immediately and equals 25\% of current cash; the previous purchase schedule spread the same input quantity over monthly payments. The agreement reduces the input price per unit by 15\% and covers 9 months of planned input use.
\end{quote}
\noindent\textit{Evaluative arm, clause order B} \hfill {\footnotesize\texttt{6363bd150010}}
\begin{quote}\singlespacing
SmallCo has accepted a bulk purchase agreement for its main production input. Payment is due immediately and equals 25\% of current cash; the previous purchase schedule spread the same input quantity over monthly payments. The immediate cash commitment is unfavorable. The agreement reduces the input price per unit by 15\% and covers 9 months of planned input use. The lower input price is favorable.
\end{quote}

\paragraph{Theme 3. License payment and exclusivity}
\begin{description}\singlespacing
\item[F1] SmallCo has signed an exclusive license for its existing production patent.
\item[F2] The licensee will pay an upfront amount equal to 20\% of current cash.
\item[F3] The license lasts 4 years, and SmallCo cannot use the patent or grant any additional licenses during that period; products using the patent generated 15\% of last year's revenue.
\item[E1] The upfront payment is favorable.
\item[E2] The restriction on patent use is unfavorable.
\end{description}
\noindent\textit{Factual arm, clause order A} \hfill {\footnotesize\texttt{25fab9ea7cf9}}
\begin{quote}\singlespacing
SmallCo has signed an exclusive license for its existing production patent. The licensee will pay an upfront amount equal to 20\% of current cash. The license lasts 4 years, and SmallCo cannot use the patent or grant any additional licenses during that period; products using the patent generated 15\% of last year's revenue.
\end{quote}
\noindent\textit{Evaluative arm, clause order A} \hfill {\footnotesize\texttt{aec729a041eb}}
\begin{quote}\singlespacing
SmallCo has signed an exclusive license for its existing production patent. The licensee will pay an upfront amount equal to 20\% of current cash. The upfront payment is favorable. The license lasts 4 years, and SmallCo cannot use the patent or grant any additional licenses during that period; products using the patent generated 15\% of last year's revenue. The restriction on patent use is unfavorable.
\end{quote}
\noindent\textit{Factual arm, clause order B} \hfill {\footnotesize\texttt{a2c6f841ff79}}
\begin{quote}\singlespacing
SmallCo has signed an exclusive license for its existing production patent. The license lasts 4 years, and SmallCo cannot use the patent or grant any additional licenses during that period; products using the patent generated 15\% of last year's revenue. The licensee will pay an upfront amount equal to 20\% of current cash.
\end{quote}
\noindent\textit{Evaluative arm, clause order B} \hfill {\footnotesize\texttt{ad018eb20571}}
\begin{quote}\singlespacing
SmallCo has signed an exclusive license for its existing production patent. The license lasts 4 years, and SmallCo cannot use the patent or grant any additional licenses during that period; products using the patent generated 15\% of last year's revenue. The restriction on patent use is unfavorable. The licensee will pay an upfront amount equal to 20\% of current cash. The upfront payment is favorable.
\end{quote}

\paragraph{Theme 4. Operating permission and required investment}
\begin{description}\singlespacing
\item[F1] SmallCo has received an operating permit for a new production line.
\item[F2] The permit allows annual output of up to 40\% of the existing line's capacity and remains valid for 5 years.
\item[F3] Use of the permit requires a compliance investment equal to 30\% of current cash before operations begin; the announcement contains no orders for the new line.
\item[E1] The additional permitted output is favorable.
\item[E2] The required cash investment is unfavorable.
\end{description}
\noindent\textit{Factual arm, clause order A} \hfill {\footnotesize\texttt{31bd76b83407}}
\begin{quote}\singlespacing
SmallCo has received an operating permit for a new production line. The permit allows annual output of up to 40\% of the existing line's capacity and remains valid for 5 years. Use of the permit requires a compliance investment equal to 30\% of current cash before operations begin; the announcement contains no orders for the new line.
\end{quote}
\noindent\textit{Evaluative arm, clause order A} \hfill {\footnotesize\texttt{2d32415f94a7}}
\begin{quote}\singlespacing
SmallCo has received an operating permit for a new production line. The permit allows annual output of up to 40\% of the existing line's capacity and remains valid for 5 years. The additional permitted output is favorable. Use of the permit requires a compliance investment equal to 30\% of current cash before operations begin; the announcement contains no orders for the new line. The required cash investment is unfavorable.
\end{quote}
\noindent\textit{Factual arm, clause order B} \hfill {\footnotesize\texttt{3b295cc1cc8f}}
\begin{quote}\singlespacing
SmallCo has received an operating permit for a new production line. Use of the permit requires a compliance investment equal to 30\% of current cash before operations begin; the announcement contains no orders for the new line. The permit allows annual output of up to 40\% of the existing line's capacity and remains valid for 5 years.
\end{quote}
\noindent\textit{Evaluative arm, clause order B} \hfill {\footnotesize\texttt{8970f996737c}}
\begin{quote}\singlespacing
SmallCo has received an operating permit for a new production line. Use of the permit requires a compliance investment equal to 30\% of current cash before operations begin; the announcement contains no orders for the new line. The required cash investment is unfavorable. The permit allows annual output of up to 40\% of the existing line's capacity and remains valid for 5 years. The additional permitted output is favorable.
\end{quote}

\paragraph{Theme 5. Cash proceeds and sold earnings}
\begin{description}\singlespacing
\item[F1] SmallCo has agreed to sell one of its operating divisions.
\item[F2] The cash sale price equals 35\% of current total company assets.
\item[F3] The division generated 20\% of last year's company operating profit, and the agreement transfers all of its operations to the buyer at closing.
\item[E1] The cash proceeds are favorable.
\item[E2] The transfer of profit-generating operations is unfavorable.
\end{description}
\noindent\textit{Factual arm, clause order A} \hfill {\footnotesize\texttt{2a3d810a7d38}}
\begin{quote}\singlespacing
SmallCo has agreed to sell one of its operating divisions. The cash sale price equals 35\% of current total company assets. The division generated 20\% of last year's company operating profit, and the agreement transfers all of its operations to the buyer at closing.
\end{quote}
\noindent\textit{Evaluative arm, clause order A} \hfill {\footnotesize\texttt{497f8b3ab18c}}
\begin{quote}\singlespacing
SmallCo has agreed to sell one of its operating divisions. The cash sale price equals 35\% of current total company assets. The cash proceeds are favorable. The division generated 20\% of last year's company operating profit, and the agreement transfers all of its operations to the buyer at closing. The transfer of profit-generating operations is unfavorable.
\end{quote}
\noindent\textit{Factual arm, clause order B} \hfill {\footnotesize\texttt{4930353e3f28}}
\begin{quote}\singlespacing
SmallCo has agreed to sell one of its operating divisions. The division generated 20\% of last year's company operating profit, and the agreement transfers all of its operations to the buyer at closing. The cash sale price equals 35\% of current total company assets.
\end{quote}
\noindent\textit{Evaluative arm, clause order B} \hfill {\footnotesize\texttt{885c628496de}}
\begin{quote}\singlespacing
SmallCo has agreed to sell one of its operating divisions. The division generated 20\% of last year's company operating profit, and the agreement transfers all of its operations to the buyer at closing. The transfer of profit-generating operations is unfavorable. The cash sale price equals 35\% of current total company assets. The cash proceeds are favorable.
\end{quote}

\paragraph{Theme 6. Product performance and warranty obligation}
\begin{description}\singlespacing
\item[F1] SmallCo has announced a redesigned version of its existing product.
\item[F2] In the reported test, the redesigned product used 18\% less energy than the current version under the same operating conditions.
\item[F3] The announced selling price is unchanged, the warranty period increases from 12 months to 24 months, and the warranty provision per unit rises by 10\%.
\item[E1] The reported energy saving is favorable.
\item[E2] The higher warranty provision is unfavorable.
\end{description}
\noindent\textit{Factual arm, clause order A} \hfill {\footnotesize\texttt{0e41745c8d4b}}
\begin{quote}\singlespacing
SmallCo has announced a redesigned version of its existing product. In the reported test, the redesigned product used 18\% less energy than the current version under the same operating conditions. The announced selling price is unchanged, the warranty period increases from 12 months to 24 months, and the warranty provision per unit rises by 10\%.
\end{quote}
\noindent\textit{Evaluative arm, clause order A} \hfill {\footnotesize\texttt{0d335462b541}}
\begin{quote}\singlespacing
SmallCo has announced a redesigned version of its existing product. In the reported test, the redesigned product used 18\% less energy than the current version under the same operating conditions. The reported energy saving is favorable. The announced selling price is unchanged, the warranty period increases from 12 months to 24 months, and the warranty provision per unit rises by 10\%. The higher warranty provision is unfavorable.
\end{quote}
\noindent\textit{Factual arm, clause order B} \hfill {\footnotesize\texttt{8e4e31c6b62b}}
\begin{quote}\singlespacing
SmallCo has announced a redesigned version of its existing product. The announced selling price is unchanged, the warranty period increases from 12 months to 24 months, and the warranty provision per unit rises by 10\%. In the reported test, the redesigned product used 18\% less energy than the current version under the same operating conditions.
\end{quote}
\noindent\textit{Evaluative arm, clause order B} \hfill {\footnotesize\texttt{37596506ed73}}
\begin{quote}\singlespacing
SmallCo has announced a redesigned version of its existing product. The announced selling price is unchanged, the warranty period increases from 12 months to 24 months, and the warranty provision per unit rises by 10\%. The higher warranty provision is unfavorable. In the reported test, the redesigned product used 18\% less energy than the current version under the same operating conditions. The reported energy saving is favorable.
\end{quote}

\subsection{Study 2 replication texts}\label{app:study2repl}
These 24 texts are the second, independently worded realizations of the six themes. They were archived on 22 October 2026 and admitted without changes on 14 November 2026 (Appendix~\ref{app:study2}). The layout follows Appendix~\ref{app:study2texts}. The file \nolinkurl{study2/texts/REPLICATION_TEXTS_FROZEN_2026-10-22.json} contains the same strings with the SHA-256 hash of the UTF-8 bytes of each, and the first 12 hexadecimal digits are shown with each text.

\paragraph{Theme 1. Customer contract and unit margin}
\begin{description}\singlespacing
\item[F1] SmallCo has signed a supply contract that runs for the next 18 months.
\item[F2] Under the contract, committed unit orders are 30\% higher than those in the current order book.
\item[F3] The stated unit production cost remains unchanged, while the contracted unit price is 12\% lower than the current average selling price.
\item[E1] An expanded order book is favorable.
\item[E2] The agreement's lower per-unit pricing is unfavorable.
\end{description}
\noindent\textit{Factual arm, clause order A} \hfill {\footnotesize\texttt{d3e47c3c86c6}}
\begin{quote}\singlespacing
SmallCo has signed a supply contract that runs for the next 18 months. Under the contract, committed unit orders are 30\% higher than those in the current order book. The stated unit production cost remains unchanged, while the contracted unit price is 12\% lower than the current average selling price.
\end{quote}
\noindent\textit{Evaluative arm, clause order A} \hfill {\footnotesize\texttt{89bc23a21c8d}}
\begin{quote}\singlespacing
SmallCo has signed a supply contract that runs for the next 18 months. Under the contract, committed unit orders are 30\% higher than those in the current order book. An expanded order book is favorable. The stated unit production cost remains unchanged, while the contracted unit price is 12\% lower than the current average selling price. The agreement's lower per-unit pricing is unfavorable.
\end{quote}
\noindent\textit{Factual arm, clause order B} \hfill {\footnotesize\texttt{dfac1762a1e9}}
\begin{quote}\singlespacing
SmallCo has signed a supply contract that runs for the next 18 months. The stated unit production cost remains unchanged, while the contracted unit price is 12\% lower than the current average selling price. Under the contract, committed unit orders are 30\% higher than those in the current order book.
\end{quote}
\noindent\textit{Evaluative arm, clause order B} \hfill {\footnotesize\texttt{59a00d115ff8}}
\begin{quote}\singlespacing
SmallCo has signed a supply contract that runs for the next 18 months. The stated unit production cost remains unchanged, while the contracted unit price is 12\% lower than the current average selling price. The agreement's lower per-unit pricing is unfavorable. Under the contract, committed unit orders are 30\% higher than those in the current order book. An expanded order book is favorable.
\end{quote}

\paragraph{Theme 2. Supplier discount and cash commitment}
\begin{description}\singlespacing
\item[F1] SmallCo has accepted a contract to purchase its primary production input in bulk.
\item[F2] The input price per unit was reduced by 15\% by the agreement, which covers 9 months of planned input use.
\item[F3] The previous purchase schedule used monthly payments for the same input quantity, whereas the payment amounts to 25\% of current cash and must be made immediately.
\item[E1] A reduction in the input price is favorable.
\item[E2] The immediate cash outlay is unfavorable.
\end{description}
\noindent\textit{Factual arm, clause order A} \hfill {\footnotesize\texttt{49941c9afff5}}
\begin{quote}\singlespacing
SmallCo has accepted a contract to purchase its primary production input in bulk. The input price per unit was reduced by 15\% by the agreement, which covers 9 months of planned input use. The previous purchase schedule used monthly payments for the same input quantity, whereas the payment amounts to 25\% of current cash and must be made immediately.
\end{quote}
\noindent\textit{Evaluative arm, clause order A} \hfill {\footnotesize\texttt{b2080efcb338}}
\begin{quote}\singlespacing
SmallCo has accepted a contract to purchase its primary production input in bulk. The input price per unit was reduced by 15\% by the agreement, which covers 9 months of planned input use. A reduction in the input price is favorable. The previous purchase schedule used monthly payments for the same input quantity, whereas the payment amounts to 25\% of current cash and must be made immediately. The immediate cash outlay is unfavorable.
\end{quote}
\noindent\textit{Factual arm, clause order B} \hfill {\footnotesize\texttt{aa705c30d321}}
\begin{quote}\singlespacing
SmallCo has accepted a contract to purchase its primary production input in bulk. The previous purchase schedule used monthly payments for the same input quantity, whereas the payment amounts to 25\% of current cash and must be made immediately. The input price per unit was reduced by 15\% by the agreement, which covers 9 months of planned input use.
\end{quote}
\noindent\textit{Evaluative arm, clause order B} \hfill {\footnotesize\texttt{668c2a1dcbfc}}
\begin{quote}\singlespacing
SmallCo has accepted a contract to purchase its primary production input in bulk. The previous purchase schedule used monthly payments for the same input quantity, whereas the payment amounts to 25\% of current cash and must be made immediately. The immediate cash outlay is unfavorable. The input price per unit was reduced by 15\% by the agreement, which covers 9 months of planned input use. A reduction in the input price is favorable.
\end{quote}

\paragraph{Theme 3. License payment and exclusivity}
\begin{description}\singlespacing
\item[F1] SmallCo has granted another company an exclusive license to use its existing production patent.
\item[F2] An upfront payment equal to 20\% of current cash will be made by the licensee.
\item[F3] For the license's 4-year term, SmallCo may neither use the patent nor grant further licenses; products based on the patent generated 15\% of last year's revenue.
\item[E1] Being paid upfront is favorable.
\item[E2] The inability to use the patent is unfavorable.
\end{description}
\noindent\textit{Factual arm, clause order A} \hfill {\footnotesize\texttt{0e4da685063a}}
\begin{quote}\singlespacing
SmallCo has granted another company an exclusive license to use its existing production patent. An upfront payment equal to 20\% of current cash will be made by the licensee. For the license's 4-year term, SmallCo may neither use the patent nor grant further licenses; products based on the patent generated 15\% of last year's revenue.
\end{quote}
\noindent\textit{Evaluative arm, clause order A} \hfill {\footnotesize\texttt{6a2a0427cc2c}}
\begin{quote}\singlespacing
SmallCo has granted another company an exclusive license to use its existing production patent. An upfront payment equal to 20\% of current cash will be made by the licensee. Being paid upfront is favorable. For the license's 4-year term, SmallCo may neither use the patent nor grant further licenses; products based on the patent generated 15\% of last year's revenue. The inability to use the patent is unfavorable.
\end{quote}
\noindent\textit{Factual arm, clause order B} \hfill {\footnotesize\texttt{6735570aea18}}
\begin{quote}\singlespacing
SmallCo has granted another company an exclusive license to use its existing production patent. For the license's 4-year term, SmallCo may neither use the patent nor grant further licenses; products based on the patent generated 15\% of last year's revenue. An upfront payment equal to 20\% of current cash will be made by the licensee.
\end{quote}
\noindent\textit{Evaluative arm, clause order B} \hfill {\footnotesize\texttt{48c9f7465d28}}
\begin{quote}\singlespacing
SmallCo has granted another company an exclusive license to use its existing production patent. For the license's 4-year term, SmallCo may neither use the patent nor grant further licenses; products based on the patent generated 15\% of last year's revenue. The inability to use the patent is unfavorable. An upfront payment equal to 20\% of current cash will be made by the licensee. Being paid upfront is favorable.
\end{quote}

\paragraph{Theme 4. Operating permission and required investment}
\begin{description}\singlespacing
\item[F1] SmallCo has obtained authorization to operate a new production line.
\item[F2] The 5-year authorization caps the new line's annual output at 40\% of the existing line's production capacity.
\item[F3] Using the permit requires spending 30\% of current cash on compliance before production starts, whereas the announcement mentions no orders for the new line.
\item[E1] The extra permitted output is favorable.
\item[E2] The mandatory cash investment is unfavorable.
\end{description}
\noindent\textit{Factual arm, clause order A} \hfill {\footnotesize\texttt{9aa5a0d03bbf}}
\begin{quote}\singlespacing
SmallCo has obtained authorization to operate a new production line. The 5-year authorization caps the new line's annual output at 40\% of the existing line's production capacity. Using the permit requires spending 30\% of current cash on compliance before production starts, whereas the announcement mentions no orders for the new line.
\end{quote}
\noindent\textit{Evaluative arm, clause order A} \hfill {\footnotesize\texttt{40f853581331}}
\begin{quote}\singlespacing
SmallCo has obtained authorization to operate a new production line. The 5-year authorization caps the new line's annual output at 40\% of the existing line's production capacity. The extra permitted output is favorable. Using the permit requires spending 30\% of current cash on compliance before production starts, whereas the announcement mentions no orders for the new line. The mandatory cash investment is unfavorable.
\end{quote}
\noindent\textit{Factual arm, clause order B} \hfill {\footnotesize\texttt{695597a0bdda}}
\begin{quote}\singlespacing
SmallCo has obtained authorization to operate a new production line. Using the permit requires spending 30\% of current cash on compliance before production starts, whereas the announcement mentions no orders for the new line. The 5-year authorization caps the new line's annual output at 40\% of the existing line's production capacity.
\end{quote}
\noindent\textit{Evaluative arm, clause order B} \hfill {\footnotesize\texttt{d6a3c7b3eeaa}}
\begin{quote}\singlespacing
SmallCo has obtained authorization to operate a new production line. Using the permit requires spending 30\% of current cash on compliance before production starts, whereas the announcement mentions no orders for the new line. The mandatory cash investment is unfavorable. The 5-year authorization caps the new line's annual output at 40\% of the existing line's production capacity. The extra permitted output is favorable.
\end{quote}

\paragraph{Theme 5. Cash proceeds and sold earnings}
\begin{description}\singlespacing
\item[F1] SmallCo has reached an agreement to sell one of its operating divisions.
\item[F2] The cash sale price amounts to 35\% of the company's current total assets.
\item[F3] The business unit contributed 20\% of the company's operating profit last year; the buyer will take over all its operations at closing under the agreement.
\item[E1] Receiving the cash proceeds is favorable.
\item[E2] It is unfavorable to transfer profit-generating operations.
\end{description}
\noindent\textit{Factual arm, clause order A} \hfill {\footnotesize\texttt{aac6a5239391}}
\begin{quote}\singlespacing
SmallCo has reached an agreement to sell one of its operating divisions. The cash sale price amounts to 35\% of the company's current total assets. The business unit contributed 20\% of the company's operating profit last year; the buyer will take over all its operations at closing under the agreement.
\end{quote}
\noindent\textit{Evaluative arm, clause order A} \hfill {\footnotesize\texttt{8b7f30e84432}}
\begin{quote}\singlespacing
SmallCo has reached an agreement to sell one of its operating divisions. The cash sale price amounts to 35\% of the company's current total assets. Receiving the cash proceeds is favorable. The business unit contributed 20\% of the company's operating profit last year; the buyer will take over all its operations at closing under the agreement. It is unfavorable to transfer profit-generating operations.
\end{quote}
\noindent\textit{Factual arm, clause order B} \hfill {\footnotesize\texttt{6abac4991d25}}
\begin{quote}\singlespacing
SmallCo has reached an agreement to sell one of its operating divisions. The business unit contributed 20\% of the company's operating profit last year; the buyer will take over all its operations at closing under the agreement. The cash sale price amounts to 35\% of the company's current total assets.
\end{quote}
\noindent\textit{Evaluative arm, clause order B} \hfill {\footnotesize\texttt{77bfe6e91ff3}}
\begin{quote}\singlespacing
SmallCo has reached an agreement to sell one of its operating divisions. The business unit contributed 20\% of the company's operating profit last year; the buyer will take over all its operations at closing under the agreement. It is unfavorable to transfer profit-generating operations. The cash sale price amounts to 35\% of the company's current total assets. Receiving the cash proceeds is favorable.
\end{quote}

\paragraph{Theme 6. Product performance and warranty obligation}
\begin{description}\singlespacing
\item[F1] SmallCo has unveiled a redesigned version of its current product.
\item[F2] Under constant operating conditions, a reported test found that the redesigned product consumed 18\% less energy than the existing version.
\item[F3] The announced selling price remains the same, while warranty coverage extends from 12 to 24 months and the per-unit warranty provision increases by 10\%.
\item[E1] The reported energy-use reduction is favorable.
\item[E2] The warranty provision increase is unfavorable.
\end{description}
\noindent\textit{Factual arm, clause order A} \hfill {\footnotesize\texttt{157c129a5512}}
\begin{quote}\singlespacing
SmallCo has unveiled a redesigned version of its current product. Under constant operating conditions, a reported test found that the redesigned product consumed 18\% less energy than the existing version. The announced selling price remains the same, while warranty coverage extends from 12 to 24 months and the per-unit warranty provision increases by 10\%.
\end{quote}
\noindent\textit{Evaluative arm, clause order A} \hfill {\footnotesize\texttt{1221a1270e94}}
\begin{quote}\singlespacing
SmallCo has unveiled a redesigned version of its current product. Under constant operating conditions, a reported test found that the redesigned product consumed 18\% less energy than the existing version. The reported energy-use reduction is favorable. The announced selling price remains the same, while warranty coverage extends from 12 to 24 months and the per-unit warranty provision increases by 10\%. The warranty provision increase is unfavorable.
\end{quote}
\noindent\textit{Factual arm, clause order B} \hfill {\footnotesize\texttt{18ca49b60a49}}
\begin{quote}\singlespacing
SmallCo has unveiled a redesigned version of its current product. The announced selling price remains the same, while warranty coverage extends from 12 to 24 months and the per-unit warranty provision increases by 10\%. Under constant operating conditions, a reported test found that the redesigned product consumed 18\% less energy than the existing version.
\end{quote}
\noindent\textit{Evaluative arm, clause order B} \hfill {\footnotesize\texttt{d2a3473e7f2c}}
\begin{quote}\singlespacing
SmallCo has unveiled a redesigned version of its current product. The announced selling price remains the same, while warranty coverage extends from 12 to 24 months and the per-unit warranty provision increases by 10\%. The warranty provision increase is unfavorable. Under constant operating conditions, a reported test found that the redesigned product consumed 18\% less energy than the existing version. The reported energy-use reduction is favorable.
\end{quote}

\clearpage
{\singlespacing
\bibliographystyle{plainnat}
\bibliography{references}

\begin{thebibliography}{25}
\providecommand{\natexlab}[1]{#1}
\providecommand{\url}[1]{\texttt{#1}}
\expandafter\ifx\csname urlstyle\endcsname\relax
  \providecommand{\doi}[1]{doi: #1}\else
  \providecommand{\doi}{doi: \begingroup \urlstyle{rm}\Url}\fi

\bibitem[Bali et~al.(2026)Bali, Kelly, M{\"o}rke, and Rahman]{bali}
Turan~G. Bali, Bryan~T. Kelly, Mathis M{\"o}rke, and Jamil Rahman.
\newblock {Machine Forecast Disagreement}.
\newblock \emph{Review of Financial Studies}, 2026.
\newblock \doi{10.1093/rfs/hhag042}.
\newblock Advance article hhag042.

\bibitem[Banerjee et~al.(2018)Banerjee, Davis, and Gondhi]{banerjee}
Snehal Banerjee, Jesse Davis, and Naveen Gondhi.
\newblock {When Transparency Improves, Must Prices Reflect Fundamentals
  Better?}
\newblock \emph{Review of Financial Studies}, 31\penalty0 (6):\penalty0
  2377--2414, 2018.
\newblock \doi{10.1093/rfs/hhy034}.

\bibitem[Bommasani et~al.(2022)Bommasani, Creel, Kumar, Jurafsky, and
  Liang]{bommasani}
Rishi Bommasani, Kathleen~A. Creel, Ananya Kumar, Dan Jurafsky, and Percy
  Liang.
\newblock {Picking on the Same Person: Does Algorithmic Monoculture lead to
  Outcome Homogenization?}
\newblock In \emph{Advances in Neural Information Processing Systems},
  volume~35, 2022.
\newblock URL
  \url{https://proceedings.neurips.cc/paper_files/paper/2022/hash/17a234c91f746d9625a75cf8a8731ee2-Abstract.html}.

\bibitem[Cataneo(2026)]{cataneo}
Felipe Cataneo.
\newblock {Provider-Dependent Behavioral Asymmetries in LLM Trading Agents:
  Evidence from the Brazilian Equity Market}.
\newblock Working paper, SSRN 6776298, 2026.
\newblock URL
  \url{https://papers.ssrn.com/sol3/papers.cfm?abstract_id=6776298}.

\bibitem[Cheng et~al.(2025)Cheng, Lin, and Zhao]{cheng}
Qiang Cheng, Pengkai Lin, and Yue Zhao.
\newblock {Does Generative AI Facilitate Investor Trading? Early Evidence from
  ChatGPT Outages}.
\newblock \emph{Journal of Accounting and Economics}, 80\penalty0
  (2--3):\penalty0 101821, 2025.
\newblock \doi{10.1016/j.jacceco.2025.101821}.

\bibitem[Dou et~al.(2025)Dou, Goldstein, and Ji]{dou}
Winston~Wei Dou, Itay Goldstein, and Yan Ji.
\newblock {AI-Powered Trading, Algorithmic Collusion, and Price Efficiency}.
\newblock NBER Working Paper 34054, 2025.
\newblock URL \url{https://www.nber.org/papers/w34054}.

\bibitem[Engelberg and Parsons(2011)]{engelberg}
Joseph~E. Engelberg and Christopher~A. Parsons.
\newblock {The Causal Impact of Media in Financial Markets}.
\newblock \emph{Journal of Finance}, 66\penalty0 (1):\penalty0 67--97, 2011.
\newblock \doi{10.1111/j.1540-6261.2010.01626.x}.

\bibitem[{Financial Stability Board}(2024)]{fsb2024}
{Financial Stability Board}.
\newblock {The Financial Stability Implications of Artificial Intelligence}.
\newblock Report, 14 November 2024, 2024.
\newblock URL
  \url{https://www.fsb.org/2024/11/the-financial-stability-implications-of-artificial-intelligence/}.

\bibitem[Glasserman and Lin(2023)]{glasserman}
Paul Glasserman and Caden Lin.
\newblock {Assessing Look-Ahead Bias in Stock Return Predictions Generated by
  {GPT} Sentiment Analysis}.
\newblock Working paper, arXiv:2309.17322, 2023.
\newblock URL \url{https://arxiv.org/abs/2309.17322}.
\newblock Online version consulted September 17, 2026.

\bibitem[Goldstein and Yang(2015)]{goldstein}
Itay Goldstein and Liyan Yang.
\newblock {Information Diversity and Complementarities in Trading and
  Information Acquisition}.
\newblock \emph{Journal of Finance}, 70\penalty0 (4):\penalty0 1723--1765,
  2015.
\newblock \doi{10.1111/jofi.12226}.

\bibitem[Gu et~al.(2020)Gu, Kelly, and Xiu]{gu}
Shihao Gu, Bryan Kelly, and Dacheng Xiu.
\newblock {Empirical Asset Pricing via Machine Learning}.
\newblock \emph{Review of Financial Studies}, 33\penalty0 (5):\penalty0
  2223--2273, 2020.
\newblock \doi{10.1093/rfs/hhaa009}.

\bibitem[Hansen and Lee(2025)]{hansenlee}
Anne~Lundgaard Hansen and Seung~Jung Lee.
\newblock {Financial Stability Implications of Generative {AI}: Taming the
  Animal Spirits}.
\newblock Finance and Economics Discussion Series 2025-090, Board of Governors
  of the Federal Reserve System, 2025.
\newblock URL \url{https://arxiv.org/abs/2510.01451}.

\bibitem[Hendershott et~al.(2011)Hendershott, Jones, and Menkveld]{hendershott}
Terrence Hendershott, Charles~M. Jones, and Albert~J. Menkveld.
\newblock {Does Algorithmic Trading Improve Liquidity?}
\newblock \emph{Journal of Finance}, 66\penalty0 (1):\penalty0 1--33, 2011.
\newblock \doi{10.1111/j.1540-6261.2010.01624.x}.

\bibitem[Henning et~al.(2025)Henning, Ojha, Spoon, Han, and Camerer]{henning}
Thomas Henning, Siddhartha~M. Ojha, Ross Spoon, Jiatong Han, and Colin~F.
  Camerer.
\newblock {LLM Agents Do Not Replicate Human Market Traders: Evidence From
  Experimental Finance}.
\newblock Working paper, arXiv:2502.15800, 2025.
\newblock URL \url{https://arxiv.org/abs/2502.15800}.
\newblock Version 3, October 11, 2025.

\bibitem[Kleinberg and Raghavan(2021)]{kleinberg}
Jon Kleinberg and Manish Raghavan.
\newblock {Algorithmic Monoculture and Social Welfare}.
\newblock \emph{Proceedings of the National Academy of Sciences}, 118\penalty0
  (22):\penalty0 e2018340118, 2021.
\newblock \doi{10.1073/pnas.2018340118}.

\bibitem[Lee and Park(2026)]{leepark}
Dohun Lee and Hyunwoo Park.
\newblock {Market Signal Injection: Adversarial Context Manipulation of LLM
  Pricing Agents}.
\newblock Working paper, arXiv:2609.18357, 2026.
\newblock URL \url{https://arxiv.org/abs/2609.18357}.
\newblock Accepted to the FinNLP 2026 Workshop at EMNLP 2026.

\bibitem[Lopez-Lira(2025)]{lopeztrade}
Alejandro Lopez-Lira.
\newblock {Can Large Language Models Trade? Testing Financial Theories with
  {LLM} Agents in Market Simulations}.
\newblock Working paper, arXiv:2504.10789, 2025.
\newblock URL \url{https://arxiv.org/abs/2504.10789}.
\newblock Online version consulted September 17, 2026.

\bibitem[Lopez-Lira and Tang(2026)]{lopez2026}
Alejandro Lopez-Lira and Yuehua Tang.
\newblock {Can {ChatGPT} Forecast Stock Price Movements? Return Predictability
  and Large Language Models}.
\newblock \emph{Journal of Financial Economics}, 184:\penalty0 104335, 2026.
\newblock \doi{10.1016/j.jfineco.2026.104335}.

\bibitem[Ouyang and Sui(2026)]{ouyangsui}
Shumiao Ouyang and Pengfei Sui.
\newblock {Dissecting AI Trading: Behavioral Finance and Market Bubbles}.
\newblock Working paper, arXiv:2604.18373, 2026.
\newblock URL \url{https://arxiv.org/abs/2604.18373}.

\bibitem[Ross et~al.(2026)Ross, So, De~Simone, Pozniak, and Lo]{rosslo}
Jillian Ross, Eric So, Zoe De~Simone, Charles Pozniak, and Andrew~W. Lo.
\newblock {Why Better Models Can Create Riskier Systems: Evidence from {LLM}
  Agents in Financial Markets}.
\newblock Working paper, arXiv:2609.04373, 2026.
\newblock URL \url{https://arxiv.org/abs/2609.04373}.

\bibitem[Saxena et~al.(2026)Saxena, Pangallo, Hommes, Caccioli, and del
  Rio-Chanona]{saxena}
Maxime Saxena, Marco Pangallo, Cars Hommes, Fabio Caccioli, and R.~Maria del
  Rio-Chanona.
\newblock {Machine Spirits: Speculation and Adaptation of {LLM} Agents in Asset
  Markets}.
\newblock Working paper, arXiv:2604.18602, 2026.
\newblock URL \url{https://arxiv.org/abs/2604.18602}.
\newblock Version 2, April 29, 2026.

\bibitem[Struski et~al.(2026)Struski, Swistak, Okulska, and Biecek]{struski}
Pawe{\l} Struski, Jakub Swistak, Inez Okulska, and Przemys{\l}aw Biecek.
\newblock {Competitive Market Behavior of {LLMs}}.
\newblock Working paper, arXiv:2609.02580, 2026.
\newblock URL \url{https://arxiv.org/abs/2609.02580}.

\bibitem[Tetlock(2011)]{tetlock}
Paul~C. Tetlock.
\newblock {All the News That's Fit to Reprint: Do Investors React to Stale
  Information?}
\newblock \emph{Review of Financial Studies}, 24\penalty0 (5):\penalty0
  1481--1512, 2011.
\newblock URL
  \url{https://papers.ssrn.com/sol3/papers.cfm?abstract_id=1018221}.

\bibitem[Wang and Zhang(2026)]{wangzhang}
Zice Wang and Zhenyu Zhang.
\newblock {Framing Effects in Independent-Agent Large Language Models: A
  Cross-Family Behavioral Analysis}.
\newblock Working paper, arXiv:2603.19282, 2026.
\newblock URL \url{https://arxiv.org/abs/2603.19282}.
\newblock Version 2, April 6, 2026.

\bibitem[Waudby-Smith and Ramdas(2020)]{betting}
Ian Waudby-Smith and Aaditya Ramdas.
\newblock {Estimating means of bounded random variables by betting}.
\newblock arXiv:2010.09686, 2020.
\newblock URL \url{https://arxiv.org/abs/2010.09686}.
\newblock Methodological reference; arXiv record consulted September 22, 2026.

\end{thebibliography}
}
\end{document}